\IfFileExists 
    {\jobname.cfg}
    {\input {\jobname.cfg}\def\OPTIONS{}}
    {\def\OPTIONS{[runningheads,a4paper]}\def \CLASS{llncs}}

\expandafter
\documentclass\OPTIONS{\CLASS}

\usepackage{gadts}

\title{\GADTs meet subtyping}

\author{Gabriel Scherer \and Didier R\'emy}

\RR{}{
\institute 
   {INRIA, Rocquencourt\thanks
    {Part of this work has been done at \textsc{IRILL}.}}
}

\csname cmdargs\endcsname
\InputIfFileExists{customize.sty}{}

\providecommand {\Final}{\False}

\Esop
  {\setanswer{hide,later}}
  {\setanswer{inline}}

\Final 
  {\UNXXX}
  {}

\begin{document}

\begin{version}{\RR}
\selectlanguage{francais}
  \makeRR
\selectlanguage{english}
\end{version}
\begin{version}{\Not\RR}

\maketitle 

\begin{abstract}{}
  \RR{\noindent}{}
  While generalized algebraic datatypes~(\GADTs) are now considered
  well-understood, adding them to a language with a notion of
  subtyping comes with a few surprises. What does it mean for a \GADT
  parameter to be covariant?
  The answer turns out to be quite subtle. It involves fine-grained
  properties of the subtyping relation that raise interesting design
  questions.
  We allow variance annotations in \GADT definitions, study their
  soundness, and present a sound and complete algorithm to check them.
  Our work may be applied to real-world ML-like languages with
  explicit subtyping such as OCaml, or to languages with general
  subtyping constraints.
\end{abstract}

\end{version}

\section*{Introduction}


In languages that have a notion of subtyping, the interface of
parametrized types usually specifies a \emph{variance}. It defines the
subtyping relation between two instances of a parametrized type from
the subtyping relations that hold between their parameters. For
example, the type $\ty{\alpha}{list}$ of immutable lists is expected
to be \emph{covariant}: we wish $\ty{\sigma}{list} \leq
\ty{\sigma'}{list}$ as soon as $\sigma \leq \sigma'$.

Variance is essential in languages with parametric polymorphism whose
programming idioms rely on subtyping, in particular object-oriented
languages, or languages with structural datatypes such as extensible
records and variants, dependently typed languages with inductive types
(to represent positivity requirements), or  additional
information in types such as permissions, effects, \etc. A last reason
to care about variance is its use in the \emph{relaxed value
  restriction}~\cite{relaxing-the-value-restriction}: while
a possibly-effectful expression, also called an \emph{expansive
  expression}, cannot be soundly generalized in ML---unless some
sophisticated enhancement of the type system keeps track of effectful
expressions---it is always sound to generalize type variables that
only appear in covariant positions, as they may not classify mutable
data.
\begin{version}{\Not\Esop}
  This relaxation uses an intuitive subtyping argument: all
  occurrences of such type variables can be specialized to $\bot$, and
  anytime later, all covariant occurrences of the same variable
  (which are now $\bot$) can be simultaneously replaced by the same
  arbitrary type $\tau$, which is always a supertype of $\bot$.  This
  relaxation of the value-restriction is implemented in OCaml, where
  it is surprisingly useful.
\end{version}
Therefore, it is important for extensions of type definitions, such as
generalized algebraic datatypes~(\GADTs), to support it as well through
a clear and expressive definition of parameter covariance.

For example, consider the following \GADT of well-typed expressions:
\label{exp-example}
\begin{lstlisting}[xleftmargin=4em]
type $\ty{\vplus\alpha}{exp}$ =
  | Val : $\alpha \to \ty{\alpha}{exp}$
  | Int : $\tyc{int} \to \ty{\tyc{int}}{exp}$
  | Thunk : $\forall \beta.\,\ty{\beta}{exp} * (\beta \to \alpha) \to \ty{\alpha}{exp}$
  | Prod : $\forall \beta \gamma .\,\ty{\beta}{exp} * \ty{\gamma}{exp} \to \ty{(\beta*\gamma)}{exp}$
\end{lstlisting}
Is it safe to say that $\tyc{exp}$ is covariant in its type parameter? It
turns out that, using the subtyping relation of the OCaml type system, the
answer is ``yes''. But, surprisingly to us, in a type system with a top type
$\top$, the answer would be ``no''.  We introduce this example in details in
\S\ref{sec/examples}---and present some interesting counter-examples of
incorrect variance annotations.

Verifying variance annotations for simple algebraic datatypes is
straightforward: it suffices to check that covariant type variables appear
only positively and contravariant variables only negatively in the
types of the arguments of the datatype constructors. \GADTs can be
formalized as extensions of datatypes where constructors have typed
arguments, but also a set of existential variables and equality
constraints. Then, the simple check of algebraic datatypes apparently
becomes a searching problem: witnesses for existentials must be found
so as to satisfy the equality constraints. That is, there is a natural
correctness criterion (already present in previous work); however, it is
expressed in a ``semantic'' form that is not suitable for a simple
implementation in a type checker. 
We present this semantic criterion
in~\S\ref{sec/formal_setting} after reviewing the formal framework
of variance-based subtyping. 

The main contribution of our work, described in
\S\ref{sec/checking_variance_of_gadts}, is to develop a syntactic criterion
that ensures the semantics criterion. 
Our solution extends the simple
check of algebraic datatypes in a non-obvious way by introducing two
new notions. First, \emph{upward and downward-closure} of type
constructors explains how to check that a single equality constraint
is still satisfiable in presence of variance (but also raises interesting
design issues for the subtyping relation). Second, \emph{zipping} explains
when witnesses exist for existential variables, that is, when multiple
constraints using the same existential may soundly be used without
interfering with each other. These two properties are combined into a new
syntactic 
judgment of \emph{decomposability} that is central to our syntactic
criterion. We prove that our syntactic criterion is
sound and complete with respect to the semantic criterion. The
proof of soundness is relatively direct, but completeness is much
harder.

We discuss the implication of our results in~\S\ref{sec/discussion}, 
in particular the notion of upward and downward-closure properties
of type constructors, on the design of a subtyping
relation. We  also contrast this approach, motivated by the needs
of a language of a ML family, with a different and mostly orthogonal
approach taken by existing object-oriented languages, namely \csharp
and Scala, where a natural notion of \GADTs involves subtyping
constraints, rather than equality constraints. We can re-evaluate our
syntactic criterion in this setting: it is still sound, but the
question of completeness is left open.

\medskip In summary, we propose a syntactic criterion for checking the
soundness of variance annotations of \GADTs with equality constraints
in a language with subtyping. Our work is directly applicable to the
OCaml language, but our approach can also be transposed to languages
with general subtyping constraints, and raises interesting design
questions.
\begin{version}{\Esop}
A long version of the present article, containing the
detailed proofs and additional details and discussion, is available online~\cite
{Scherer-Remy:gadts-meet-subtyping@long2012}.
\end{version}

\section{Examples}
\label{sec/examples}

Let us first explain why it is reasonable to say that $\ty\alpha{exp}$ is
covariant. Informally, if we are able to coerce
a value of type  $\alpha$ into one of type $\alpha'$ (we write $(\mc{v} :> \alpha')$ to
explicitly cast a value $\mc{v}$ of type $\alpha$ to a value of type
$\alpha'$), then we are also able to transform a value of type $\ty\alpha{exp}$
into one of type $\ty{\alpha'}{exp}$. Here is some pseudo-code\footnote{The
  variables $\beta'$ and $\gamma'$ of the \code{Prod} case are never
  really defined, only justified at the meta-level, making this code
  only an informal sketch.} for the coercion function:
\begin{lstlisting}[xleftmargin=4em]
let coerce : $\ty\alpha{exp} \to \ty{\alpha'}{exp}$ = function
  | Val (v : $\alpha$) -> Val (v :> $\alpha'$)
  | Int n -> Int n
  | Thunk $\beta$ (b : $\beta$ exp) (f : $\beta \to \alpha$) ->
    Thunk $\beta$ b (fun x -> (f x :> $\alpha'$))
  | Prod $\beta$ $\gamma$ ((b, c) : $\ty\beta{exp} * \ty\gamma{exp}$) ->
    (* if $\beta * \gamma \leq \alpha'$, then $\alpha'$ is of the form $
        \beta' * \gamma'$
        with $\beta \leq \beta'$ and $\gamma \leq \gamma'$ *)
    Prod $\beta'$ $\gamma'$ ((b :> $\ty{\beta'}{exp}$), (c :> $\ty{\gamma'}{exp}$))
\end{lstlisting}
In the $\mc{Prod}$ case, we make an informal use of something we know
about the OCaml type system: the supertypes of a tuple are all
tuples. By entering the branch, we gain the knowledge that $\alpha$
must be equal to some type of the form $\beta * \gamma$. So from
$\alpha \leq \alpha'$ we know that
$\beta * \gamma \leq \alpha'$. Therefore, $\alpha'$ must itself be
a pair of the form $\beta' * \gamma'$. By covariance of the product,
we deduce that $\beta \leq \beta'$ and $\gamma \leq \gamma'$. We may 
thus conclude by casting at types $\ty{\beta'}{exp}$ and
$\ty{\gamma'}{exp}$, recursively.

Similarly, in the $\mc{Int}$ case, we know that $\alpha$ must be an
$\tyc{int}$ and therefore an $\ty{\tyc{int}}{exp}$ is returned. This is
because we know that, in OCaml, no type is above $\tyc{int}$: if $\tyc{int}
\leq \tau$, then $\tau$ must be $\tyc{int}$. 

What we use in both cases is reasoning of the form\footnote {\let \T
  T\let \S \sigma\relax We write $\app \T\bb$ for a type expression
  $\T$ that may contain free occurrences of variables $\bb$ and $\app
  \T {\bar\S}$ for the simultaneous substitution of $\bar\S$ for $\bb$
  in $\T$.}:
``if  
$\app{T}{\bb} \leq \alpha'$, then I know that $\alpha'$ is of the form
$\app{T}{\bb'}$ for some $\bb'$''. We call this an \emph{upward
  closure} property: when we ``go up'' from a $\app{T}\bb$, we only
find types that also have the structure of $T$. Similarly, for
contravariant parameters, we would need a \emph{downward closure}
property: $T$ is downward-closed if $\app{T}\bb \geq \alpha'$ entails
that $\alpha'$ is of the form $\app{T}{\bb'}$.

Before studying a more troubling example, we  define the classic
equality type $\ty{(\alpha, \beta)}{eq}$ and the corresponding
casting function
$\mc{cast}: \forall \alpha \beta. \ty{(\alpha, \beta)}{eq} \to \alpha \to
\beta$: 
\begin{version}{\Not\Esop}
\begin{lstlisting}
  type ($\alpha$, $\beta$) eq =
    | Refl : $\forall \gamma$. ($\gamma$, $\gamma$) eq
\end{lstlisting}
\begin{lstlisting}
  let cast (eqab : $\ty{(\alpha, \beta)}{eq}$) : $\alpha \to \beta$ =
    match eqab with
      | Refl -> (fun x -> x)
\end{lstlisting}
\end{version}
\begin{version}{\Esop}
\begin{lstlisting}[xleftmargin=1em]
type ($\alpha$, $\beta$) eq =                         $\hspace{5.8em}$let cast r =
  | Refl : $\forall \gamma$. ($\gamma$, $\gamma$) eq  $\hspace{2.3em}$  match r with Refl -> (fun x -> x)
\end{lstlisting}
\end{version}
Notice that it would be unsound\footnote{This counterexample is due to
  Jeremy Yallop.} to define $\tyc{eq}$ as covariant, even in only one
parameter. For example, if we had
$\mc{type}~\ty{(\vplus\alpha,\veq\beta)}{eq}$, from any $\sigma \leq
\tau$, we could subtype $\ty{(\sigma,\sigma)}{eq}$ into
$\ty{(\tau,\sigma)}{eq}$, allowing a cast from any value of type $\tau$
back into one of type $\sigma$, which is unsound in general.

As a counter-example, the following declaration is incorrect: the type
$\ty{\alpha}{bad}$ cannot be declared covariant.
\begin{lstlisting}
  type $\vplus\alpha$ bad =
    | K : < m : int > $\to$ < m : int > bad
  let v = (K (object method m = 1 end) :> < > bad)
\end{lstlisting}
This declaration uses the OCaml object type \code{< m : int >}, which
qualifies objects having a method \code{m} returning an integer. It is
a subtype of object types with fewer methods, in this case the empty
object type \code{< >}, so the alleged covariance of \code{bad}, if
accepted by the compiler, would allow us to cast a value of type
\code{< m : int > bad} into one of type \code{< > bad} and thus have 
the above value \code v of type \code {<> bad}. However, if such a value \code v
existed,  we could produce an equality witness
\code{(< >, <m : int>) eq} that allows to cast any
empty object of type \code{< >} into an object of type
\code{< m : int >}, but this is unsound, of course!

\begin{lstlisting}
  let get_eq : $\ty\alpha{bad} \to$ ($\alpha$, < m : int >) eq = function
    | K _ -> Refl      $\qquad$ (* locally $\alpha =\,$< m : int > *)
  let wrong : < > -> < m : int > =
    let eq : (< >, < m : int >) eq = get_eq v in cast eq
\end{lstlisting}
It is possible to reproduce this example using a different feature of
the OCaml type system named \emph{private type
  abbreviation}\footnote{This counterexample is due to Jacques
  Garrigue.}: a module using a type $\mc{type}~\tyc{s}~\mc{=}~\tau$
\emph{internally} may describe its interface as
$\mc{type}~\tyc{s}~\mc{=}~\mc{private}~\tau$. This is a compromise
between a type abbreviation and an abstract type: it is possible to
cast a value of type $\tyc{s}$ into one of type $\tau$, but not,
conversely, to construct a value of type $\tyc s$ from one of type
$\tau$. In other words, $\tyc{s}$ is a strict subtype of $\tau$: we
have $\tyc{s} \leq \tau$ but not $ \tyc{s} \geq\tau$. Take for example
\code{type file_descr = private int}: this semi-abstraction is useful
to enforce invariants by restricting the construction of values of
type \code{file_descr}, while allowing users to conveniently and
efficiently destruct them for inspection at type \code{int}.
\begin{version}{\Esop}
  Using an unsound but quite innocent-looking covariant \GADT
  datatype, one is able to construct a function to cast any integer
  into a \code{file_descr}, which defeats the purpose of this
  abstraction---see the extended version of this article for the full
  example.
\end{version}

\begin{version}{\Not\Esop}
Unsound \GADT covariance declarations would defeat the purpose of such
private types: as soon as the user gets one element of the private
type, she could forge values of this type, as illustrated by the code
below.
\begin{lstlisting}
  module M = struct
    type file_descr = int
    let stdin = 0
    let open = ...
  end : sig
    type file_descr = private int
    val stdin : file_descr
    val open : string -> (file_descr, error) sum
  end

  type $\vplus\alpha$ s =
    | K : priv -> M.file_descr s

  let get_eq : $\ty\alpha{s}$ -> ($\alpha$, M.file_descr) eq = function
    | K _ -> Refl

  let forge : int -> M.file_descr =
    fun (x : $\tyc{int}$) -> cast (get_eq p) M.stdin
\end{lstlisting}
\end{version}

The difference between the former, correct \code{Prod} case and those
two latter situations with unsound variance is the notion of upward
closure. The types $\alpha * \beta$ and $\tyc{int}$ used in the
correct example were upward-closed. On the contrary, the private type
\Esop{\code{file_descr}}{\code{M.file_descr}} has a distinct supertype
\code{int}, and similarly, the object type \code { < m:int > } has
a supertype \code {< >} with a different structure (no method
\code m).

\begin{version}{\Esop}
  Finally, the need for covariance of $\ty{\alpha}{exp}$ can be justified
  either by applications using subtyping on data (for example object types or
  polymorphic variants), or by the relaxed value restriction.
  If we used the \texttt{Thunk} constructor to delay a computation
  returning an object of type \code{< m : int >}, that is itself of type
  {\code{< m : int > exp}}, we may need to see it as a computation
  returning the empty object \code{< >}.
  We could also wish to define an abstract interface through a module
  boundary that would not expose any implementation detail about the
  datatype; for example, using \texttt{Product} to implement
  a \texttt{list} interface.
  \begin{lstlisting}[xleftmargin=2em]
  module Exp : sig
    type $\alpha$ exp                   
    val inj : $\alpha$ -> $\alpha$ exp  
    val pair : $\alpha$ exp -> $\beta$ exp -> $(\alpha * \beta)$ exp
    val fst : $(\alpha * \beta)$ exp -> $\alpha$ exp
  end
  \end{lstlisting}
  What would then be the type of \texttt{Exp.inj []}? In presence of
  the value restriction, this application cannot be generalized, and
  we get a weak polymorphic type \texttt{$?\alpha$ list Exp.exp} for some
  non-generalized inference variable \texttt{$?\alpha$}. If we change
  the interface to express that \texttt{Exp.exp} is covariant, then we
  get the expected polymorphic type
  $\forall \alpha. \ty{\ty{\alpha}{list}}{Exp.exp}$.
\end{version}
\begin{version}{\Not\Esop}
\subsection*{On the importance of variance annotations}

Being able to specify the variance of a parametrized datatype is
important at abstraction boundaries: one may wish to define a program
component relying on an \emph{abstract} type, but still make certain
subtyping assumptions on this type. Variance assignments provide
a framework to specify such a semantic interface with respect to
subtyping. When this abstract type dependency is provided by an
encapsulated implementation, the system must check that the provided
implementation indeed matches the claimed variance properties.

Assume the user specifies an abstract type
\begin{lstlisting}[xleftmargin=4em]
module type S = sig
  type $(+\alpha)$ collection
  val empty : unit -> $\alpha$ collection
  val app : $\alpha$ collection -> $\alpha$ collection -> $\alpha$ collection
end
\end{lstlisting}
and then implements it with linked lists
\begin{lstlisting}[xleftmargin=4em]
module C : S = struct
  type $+\alpha$ collection =
    | Nil of $\tyc{unit}$
    | Cons of $\alpha\,*\,\ty{\alpha}{collection}$
  let empty () = Nil ()
end
\end{lstlisting}
The type-checker will accept this implementation, as it has the
specified variance. On the contrary,
\begin{lstlisting}[xleftmargin=4em]
type $+\alpha$ collection = $\ty{(\ty{\alpha}{list})}{ref}$  
let empty () = ref []
\end{lstlisting}
would be rejected, as $\tyc{ref}$ is invariant.
In the following definition:
\begin{lstlisting}[xleftmargin=4em]
let nil = C.empty ()
\end{lstlisting}
the right hand-side is not a value, and is therefore not generalized
in presence of the value restriction; we get a monomorphic type,
$\ty{?\alpha}{t}$, where $?\alpha$ is a yet-undetermined type
variable. The relaxed value restriction
\cite{relaxing-the-value-restriction} indicates that it is sound to
generalize $?\alpha$, as it only appears in covariant
positions. Informally, one may unify $?\alpha$ with $\bot$, add an
innocuous quantification over $\alpha$, and then generalize
$\forall \alpha. \ty{\bot}{t}$ into $\forall \alpha. \ty{\alpha}{t}$
by covariance---assuming a lifting of subtyping to polymorphic type
schemes.

The definition of \code{nil} will therefore get generalized in
presence of the relaxed value restriction, which would not be the case
if the interface \texttt S had specified an invariant type.
\end{version}

\section{A formal setting}\label{sec/formal_setting}

\subsection{The subtyping relation}

Ground types consist of base type $\tyc q$, types $\tau\; \tyc p$,
function types $\tau_1 \to \tau_2$, product
types $\tau_1 * \tau_2$, and a set of algebraic datatypes
$\ty{\bs}{t}$.  We also write $\sigma$ and $\rho$ for types,
 $\bar \sigma$ for a sequence of types
$(\sigma_i)_\iI$, and we use prefix notation for datatype parameters, as
is the usage in ML. Datatypes may be user-defined by toplevel
declarations of the form:
$$
\type \Gva\t = \seqi\mid{\Kc[\i] \of \app{\tau^\i}\ba}
$$
This is a disjoint sum: the constructors $\mc K_c$ represent all 
possible cases and each type
$\app {\tau^c} \ba$ is the domain of the constructor
$\mc K_c$. Applying $\mc K_c$ to an argument $e$ of a corresponding ground
type $\app{\tau}{\bs}$ constructs a term of type $\ty{\bs}{t}$. Values
of this type are deconstructed using pattern matching clauses of the
form $\mc{K}_c~x \to e$, one for each constructor.

The sequence $\Gva$ is a binding list of type variables $\alpha_i$
along with their \emph{variance annotation} $v_i$. 
Variances range in the set $\set{\vplus, \vminus, \veq, \virr}$. We may
associate  a relation $\rel{\prec_v}$ between types to each variance
$v$:
\begin{itemize}
\item $\prec_{\vplus}$ is the \emph{covariant} relation $\rel\leq$;
\item $\prec_{\vminus}$ is the \emph{contravariant} relation $\rel\geq$,
the symmetric of $\rel\leq$;
\item $\prec_{\veq}$ is the \emph{invariant} relation $\rel{=}$ defined as
the intersection of $\rel\leq$ and $\rel\geq$;
\item $\prec_{\virr}$ is the \emph{irrelevant} relation $\rel\Join$, \ie.
  the full relation such that $\sigma \Join \tau$ holds for all types
  $\sigma$ and $\tau$.
\end{itemize}
Given a reflexive transitive relation $\rel\leqslant$ on base
types, the subtyping relation on ground types $\rel\leq$ is defined by the
inference rules of Figure~\ref {fig/subtyping}, which, in particular, give
their meaning to the variance annotations $\Gva$. The judgment
$\mc{type}~\ty{\Gva}{t}$ simply means that the type constructor
$\tyc{t}$ has been previously defined with the variance annotation
$\Gva$. 
\begin{mathparfig}{fig/subtyping}{Subtyping relation}
\inferrule[sub-Refl]{ }
  {\sigma \leq \sigma}

\inferrule[sub-Trans]
  {\sigma_1 \leq \sigma_2 \\ \sigma_2 \leq \sigma_3}
  {\sigma_1 \leq \sigma_3}

\inferrule[sub-Fun]
  {\sigma \geq \sigma' \\ \tau \leq \tau'}
  {\sigma \to \tau \leq \sigma' \to \tau'}
  \quad \and

\inferrule[sub-Prod]
  {\sigma \leq \sigma'\\ \tau \leq \tau'}
  {\sigma * \tau \leq \sigma' * \tau'}

\inferrule[sub-Constr]
  {\mc{type}~\ty{\Gva}{t} \\
   \forall i,\,\sigma_i \prec_{v_i} \sigma'_i}
  {\ty{\bs}{t} \leq \ty{\bs'}{t}}

\inferrule[sub-P]
  {\sigma \leq \sigma'}
  {\ty{\sigma}{p} \leq \ty{\sigma'}{p}}

\inferrule[sub-PQ]{ }
  {\ty{\sigma}{p} \leq \tyc{q}}
\end{mathparfig}
Notice that the rules for arrow and product types, \Rule{sub-Fun} and
\Rule{sub-Prod}, can be subsumed by the rule for datatypes
\Rule{sub-Constr}. Indeed, one can consider them as special datatypes
(with a specific dynamic semantics) of variance $(\vminus, \vplus)$
and $(\vplus, \vplus)$, respectively. For this reason, the following
definitions will not explicitly detail the cases for arrows and
products.

The rules \Rule{sub-P} and \Rule{sub-PQ} were added for the explicit
purpose of introducing some amount of non-atomic subtyping in our
relation. For two fixed type constructors $\tyc{p}$ (unary) and
$\tyc{q}$ (nullary), we have $\ty{\sigma}{p} \leq \tyc{q}$ for any
$\sigma$. Note that $\tyc{q}$ is not a top type as it is not above all
types, only above the $\ty{\sigma}{p}$. Of course, we could add other
such type constructors, but those are enough to make the system
interesting and representative of complex subtype relation.

\begin{version}{\Not\Esop}
Finally, it is routine to show that the rules for reflexivity and
transitivity are admissible, by pushing them up in the derivation
until the base cases $\tyc{b} \leqslant \tyc{c}$, where they can be
removed as $\rel\leqslant$ is assumed to be reflexive and transitive. 
Removing reflexivity and transitivity provides us with an equivalent
syntax-directed judgment having powerful inversion principles: if
$\ty \bs t \leq \ty {\bs'} t$ and
$\mc{type}~\ty{\Gva}{t}$, then one can deduce that for each
$i$, $\sigma_i \prec_{v_i} \sigma'_i$.\label{subtyping_inversion} 
\end{version}
\begin{version}{\Esop}
  As usual in subtyping systems, we could reformulate our judgment in
  a syntax-directed way, to prove that it admits good inversion
  properties: if $\ty \bs t \leq \ty {\bs'} t$ and
  $\mc{type}~\ty{\Gva}{t}$, then one can deduce that for each $i$,
  $\sigma_i \prec_{v_i} \sigma'_i$.\label{subtyping_inversion}
\end{version}

The non-atomic rule \Rule{sub-PQ} ensures that our subtyping relation
is not ``too structured'' and is a meaningful choice for a formal
study applicable to real-world languages with possibly top or bottom
types, private types, record width subtyping, etc. In particular, the
type constructor $\tyc{p}$ is \emph{not} upward-closed (and conversely
$\tyc{q}$ is not downward-closed), as used informally in the examples
and defined for arbitrary variances in the following way:
\begin{definition}[Constructor closure]
\label{def/v-closed}
A type constructor $\ty{\ba}{t}$ is $v$-closed if, for any type sequence
$\bs$ and type $\tau$ such that $\ty{\bs}{t} \prec_v \tau$ hold, then
$\tau$ is necessarily equal to $\ty{\bs'}{t}$ for some $\bs'$.
\end{definition}

\subsection{The algebra of variances}

If we know that $\ty{\bs}{t} \leq \ty{\bs'}{t}$, that is $\ty{\bs}{t}
\prec_{\vplus} \ty{\bs'}{t}$, and the constructor $\tyc{t}$ has variable
$\bar{v\alpha}$, an inversion principle tells us that
$\sigma_i \prec_{v_i} \sigma'_i$  for each $i$. But what if we only know
$\ty \bs t \prec_{u} \ty {\bs'} t$ for some variance $u$ different from
$\rel\vplus$? If $u$ is $\rel\vminus$, we get the reverse relation
$\sigma_i \succ_{v_i} \sigma'_i$. If $u$ is $\rel\virr$, we get
$\sigma_i \Join \sigma'_i$, that is, nothing. This outlines
a \emph{composition} operation on variances $\varcomp u {v_i}$, such
that if $\ty \bs t \prec_{u} \ty {\bs'} t$ then $\sigma_i
\prec_{\varcomp u {v_i}} \sigma'_i$ holds. It is defined by the
table in figure~\ref{fig:composition_table}.

This operation is associative and commutative. Such an operator, and
the algebraic properties of variances explained below, have already
been used by other authors, for example
\cite{polarized-subtyping-for-sized-types}.

There is a natural order relation between \emph{variances}, which
is the \emph{coarser-than} order between the corresponding relations:
$v \leq w$ if and only if 
$\rel{\prec_v} \supseteq \rel{\prec_w}$; \ie. if and only if, for all
$\sigma$ and $\tau$, $\sigma \prec_w \tau$ implies
$\sigma \prec_v \tau$.\footnote{The reason for this order reversal is
  that the relations occur as hypotheses, in negative position, in
  definition of subtyping: if we have $v \leq w$ and
  $\mc{type}~\ty{v\alpha}{t}$, it is safe to assume
  $\mc{type}~\ty{w\alpha}{t}$, since $\sigma \prec_w \sigma'$ implies
  $\sigma \prec_v \sigma'$, which implies
  $\ty{\sigma}{t} \leq \ty{\sigma'}{t}$. One may also see it, as Abel
  notes, as an ``information order'': knowing that
  $\sigma \prec_{\vplus} \tau$ ``gives you more information'' than
  knowing that $\sigma \prec_{\virr} \tau$, therefore
  $\virr \leq \vplus$.}
This reflexive, partial order is described by the lattice diagram in
figure~\ref{fig:order_diagram}. All variances are smaller than $\veq$
and bigger than $\virr$.

\begin{figure}
\abovedisplayskip 0em\belowdisplayskip 0em
\begin{minipage}[b]{0.45\linewidth}
\[
\def \C#1{\multicolumn{1}{C}{#1}}
\begin{tabular}{R|c@{}|C|C|C|C|l}
  \varcomp v w && \C{\veq} & \C{\vplus} & \C{\vminus} & \C{\virr} & $w$
\\ 
\cline{1-1}\cline{3-7} \noalign{\vskip \doublerulesep}\cline{1-1}\cline{3-6}
  \veq    && \veq  & \veq     & \veq    & \virr \\ \cline{3-6}
  \vplus  && \veq  & \vplus   & \vminus & \virr \\ \cline{3-6}
  \vminus && \veq  & \vminus  & \vplus  & \virr \\ \cline{3-6}
  \virr   && \virr & \virr    & \virr   & \virr \\ \cline{3-6}
  v       &\multicolumn{5}{C}{}
  \end{tabular}
\]
\label{fig:composition_table}
\caption{Variance composition table}
\end{minipage}
\hspace{0.5cm}
\begin{minipage}[b]{0.45\linewidth}
\[
\xymatrix@=1.4ex{
 & \veq \ar@{-}[dr] & \\
\vplus \ar@{-}[ur] \ar@{-}[dr] & & - \\
& \virr \ar@{-}[ur] & \\
}
\]
\label{fig:order_diagram}
\caption{Variance order diagram}
\end{minipage}
\end{figure}

\begin{version}{\False}
Conversely, if we assume that $D[C[\_]]$ has variance $u$, and know
that $D[\_]$ has variance $w$, what must be the variance of $C[\_]$
alone ? For a fixed $w$, the composition operation
$v \mapsto \varcomp v w$ has a ``quasi-inverse'', which we will write
$u \mapsto \vardiv u w$. We say ``quasi'' because it does not verify
$\vardiv {(\varcomp v w)} w = v$, but only the weaker
$\vardiv {(\varcomp v w)} w \leq v$. More generally, Abel gives the
following connection:
  \[ \forall u, v, w,\quad u \leq \varcomp v w \iff \vardiv u w \leq v \]
The operation table for $\vardiv u w$ is the following:
$$
  \begin{tabular}{r|c|c|c|c|l}
    \multicolumn{1}{r}{$\vardiv u w$}
    & \multicolumn{1}{c}{$\veq$}
    & \multicolumn{1}{c}{$\vplus$}
    & \multicolumn{1}{c}{$\vminus$}
    & \multicolumn{1}{c}{$\virr$}
    & $w$
    \\ \cline{2-5}
    $\veq$    & $\veq$  & $\veq$    & $\veq$    & $\veq$ \\ \cline{2-5}
    $\vplus$  & $\virr$ & $\vplus$  & $\vminus$ & $\veq$ \\ \cline{2-5}
    $\vminus$ & $\virr$ & $\vminus$ & $\vplus$  & $\veq$ \\ \cline{2-5}
    $\virr$   & $\virr$ & $\virr$   & $\virr$   & $\veq$ \\ \cline{2-5}
    \multicolumn{1}{r}{$u$}
    & \multicolumn{1}{c}{}
    & \multicolumn{1}{c}{}
    & \multicolumn{1}{c}{}
    & \multicolumn{1}{c}{}
    & \\
  \end{tabular}
$$
This operation is associative, and further we have that
$\vardiv{(\vardiv u w)} w' = \vardiv u {(\varcomp w w')}$. It is also
monotone, but contravariant in its right argument: if $v \leq v'$
and $w \geq w'$, then $\vardiv v w \leq \vardiv {v'} {w'}$.
\end{version}

From the order lattice on variances we can define join $\vee$ and meet
$\wedge$ of variances: $v \vee w$ is the biggest variance such that
$v \vee w \leq v$ and $v \vee w \leq w$; conversely, $v \wedge w$ is
the lowest variance such that $v \leq v \wedge w$ and
$w \leq v \wedge w$. Finally, the composition operation is monotone:
if $v \leq v'$ then $\varcomp w v \leq \varcomp w {v'}$
and $\varcomp v w \leq \varcomp {v'} w$.

We often manipulate vectors $\bar{v\alpha}$ of variable
associated with variances, which correpond to the ``context'' $\Gamma$
of a type declaration. We extend our operation pairwise on those
contexts: $\Gamma \vee \Gamma'$ and $\Gamma \wedge \Gamma'$, and the
ordering between contexts $\Gamma \leq \Gamma'$. We also extend the
variance-dependent subtyping relation $\rel{\prec_v}$, which becomes
an order $\rel{\prec_\Gamma}$ between vectors of type of the same
length: $\bs \prec_{\bar{v\alpha}} \bs'$ holds when we
have $\sigma_i \prec_{v_i} \sigma'_i$ for all $i$.

\subsection{A judgment for variance of type expressions}

We define a judgment to check the variance of a type expression. Given
a context $\G$ of the form $\bar{v\alpha}$, that is, where each
variable is annotated with a variance, the judgment $\G \der \tau : v$
checks that the expression $\tau$ varies along $v$ when the variables
of $\tau$ vary along their variance in $\G$.  For example,
$(\vplus\alpha) \vdash \app{\tau}{\alpha} : \vplus$ holds when
$\app{\tau}{\alpha}$ is covariant in its variable $\alpha$.  The
inference rules for the judgment $\G \der \tau : v$ are defined on
Figure~\ref{fig/variance}.

\begin{version}{\Not\Esop}
In the context of higher-order subtyping
\cite{polarized-subtyping-for-sized-types}, where type abstractions
are first-class and annotated with a variance
($\lambda v\alpha. \tau$), it is natural to present this check as
a kind checking of the form $\G \der \tau : \kappa$, where
$\G$ is a context $\Gva$ of type variables \emph{associated with
  variances}. For example, if $\vplus\alpha \der \tau : \star$ is provable, it
is sound to consider $\alpha$ covariant in $\tau$. In the context of
a simple first-order monomorphic type calculus, this amounts to
a \emph{monotonicity check} on the type $\tau$ as defined by
\cite{csharp-generalized-constraints}.
Both approaches use judgments of a peculiar form where the
\emph{context} changes when going under a type constructor: to check
$\G \der \sigma \to \tau$, one checks $\G \der \tau$ but
$(\vardiv \G \vminus) \der \sigma$, where $\vardiv \G \vminus$
reverses all the variances in the context $\G$
(turns $\rel\vminus$ into $\rel\vplus$ and conversely). Abel gives an
elegant presentation of this inversion $/$ as an algebraic operation
on variances, a quasi-inverse such that $\vardiv u v \leq w$ if and
only if $u \leq \varcomp w v$. This context change is also reminiscent
of the \emph{context resurrection} operation of the literature on
proof irrelevance (in the style of \cite{pfenning:intextirr}
for example).
\end{version}

\begin{mathparfig}{fig/variance}{Variance assignment}
\inferrule[vc-Var]
  {w\alpha \in \G\\ w \geq v}
  {\G \der \alpha : v}

\inferrule[vc-Constr]
  {\G \der \mc{type}~\ty{\bar{w\alpha}}{t}\\
   \forall i,\ \G \der \sigma_i : \varcomp v {w_i}}
  {\G \der \ty{\bar{\sigma}}{t} : v}
\end{mathparfig}

\begin{version}{\Not\Esop}
  We chose an equivalent but more conventional style where the context
  of subderivation does not change: instead of a judgment
  $\G \der \tau$ that becomes ${(\vardiv \G u) \der \sigma}$ when
  moving to a subterm of variance $u$, we define a judgment of the
  form $\G \der \tau : v$, that evolves into
  $\G \der \sigma : (\varcomp v u)$. The two styles are equally
  expressive: our judgment $\G \der \tau : v$ holds if and only if
  $(\vardiv \G v) \der \tau$ holds in Abel's system---but we found
  that this one extends more naturally to checking decomposability, as
  will later be necessary. The inference rules for the judgment
  $\G \der \tau : v$ are defined on Figure~\ref{fig/variance}.
\end{version}
\begin{version}{\Esop}
  The parameter $v$ evolves when going into subderivations: when
  checking $\G \der \tau_1 \to \tau_2 : v$, contravariance is
  expressed by checking
  $\G \der \tau_1 : \rel{\varcomp v \vminus}$. Previous work
  (on variance as \cite{polarized-subtyping-for-sized-types} and
  \cite{csharp-generalized-constraints}, but also on irrelevance as in
  \cite{pfenning:intextirr}) used no such parameter, but modified the
  context instead, checking $\vardiv \G \vminus \der \tau_1$ for
  some ``variance cancellation'' operation $v \vardiv w$
  (see \cite{polarized-subtyping-for-sized-types} for
  a principled presentation). Our own inference rules preserve the
  same context in the whole derivation and can be more easily adapted
  to the decomposability judgment $\G \der \tau : v \To v'$ that we
  introduce in \S\ref{sec/syntactic-decomposability}.
\end{version}

\paragraph{A semantics for variance assignment}

This syntactic judgment $\Gamma \der \tau : v$ corresponds to
a semantic property about the types and context involved, which
formalizes our intuition of ``when the variables vary along $\Gamma$,
the expression $\tau$ varies along $v$''. We also give a few formal
results about this judgment.

\begin{definition}[Interpretation of the variance checking judgment]
\label{def/vc/semantics}\indent\\
We write $\sem{\G \der \tau : v}$ for the property:
$
\forall \bs, \bs', \;
    \bs \prec_{\G} \bs'
    \implies \app \tau \bs \prec_v \app \tau {\bs'}
$.
\end{definition}
\begin{lemma}[Correctness of variance checking]
  \label{lem/variance-checking-correct}\indent\\
  $\G \der \tau : v$ is provable if and only if
  $\sem{\G \der \tau : v}$ holds.
\end{lemma}
\begin{proof}{inline}
  \Case{Soundness: $\G \der \tau : v$ implies
    $\sem{\G \der \tau : v}$.} By induction on the derivation. In
  the variable case this is direct. In the $\ty{\bs}{t}$ case, for
  $\br, \br'$ such that $\br \prec_\G \br'$, we get
  $\forall i, \app{\sigma_i}{\br} \prec_{\varcomp v {w_i}} \app{\sigma_i}{\br'}$
  by inductive hypothesis, which allows to conclude, by definition of
  variance composition, that
  $\app{(\ty{\bs}{t})}{\br} \prec_v \app{(\ty{\bs}{t})}{\br'}$.

  \Case{Completeness: $\sem{\G \der \tau : v}$ implies
    $\G \der \tau : v$.} By induction on $\tau$. In the variable
  case this is again direct. In the $\ty{\bs}{t}$ case, given
  $\br \prec_\G \br'$ such that
  $\app{(\ty{\bs}{t})}{\br} \prec_v \app{(\ty{\bs}{t})}{\br'}$ we 
  deduce
  $\app{\sigma_i}{\br} \prec_{\varcomp v {w_i}} \app{\sigma_i}{\br'}$
for each variable $\alpha_i$ of variance $w_i$ in $\app{\tau}{\ba}$, 
  by inversion; this allows  to inductively build the subderivations
  $\G \der \sigma_i : \varcomp v {w_i}$.
\qed
\end{proof}

\begin{lemma}[Monotonicity] 
\label{lem/monotonicity}\indent\\
  If $\G \der \tau : v$ is provable and $\G \leq \G'$ then
  $\G' \der \tau : v$ is provable.
\end{lemma}
\begin{proof}{hide}
Obvious. 
\end{proof}
\begin{version}{\Not\Esop}
\begin{lemma}
  If $\G \der \tau : v$ and $\G' \der \tau : v$ both hold,
  then $(\G \vee \G') \der \tau : v$ also holds.
\end{lemma}
\begin{proof}{hide}
Obvious
\end{proof}
\begin{corollary}[Principality] \label{principality}
  For any type $\tau$ and any variance $v$, there exists a minimal context
  $\Delta$ such that $\Delta \der \tau : v$ holds. That is, for any other
  context $\G$ such that $\G \der \tau : v$, we have $\Delta \leq
  \G$.
\end{corollary}
\end{version}

\begin{version}{\Esop}
\begin{lemma}[Principality] \label{principality}
  For any type $\tau$ and any variance $v$, there exists a minimal context
  $\Delta$ such that $\Delta \der \tau : v$ holds. That is, for any other
  context $\G$ such that $\G \der \tau : v$, we have $\Delta \leq
  \G$.
\end{lemma}
\end{version}

\begin{version}{\Not\Esop}
\paragraph{Inversion of subtyping}

We have mentioned in \ref{subtyping_inversion} the inversion
properties of our subtyping relation. From
$\ty{\bs}{t} \leq \ty{\bs'}{t}$ we can deduce subtyping relations on
the type parameters $\sigma_i, \sigma'_i$. This can be generalized to
any type expression $\app{\tau}{\ba}$:

\begin{theorem}[Inversion] \label{thm/inversion} For any type
  $\app{\tau}{\ba}$, variance $v$, and type sequences $\bs$ and
  $\bs'$, the subtyping relation
  $\app{\tau}{\bs} \prec_v \app{\tau}{\bs'}$ holds if and only if the
  judgment $\G \der \tau : v$ holds for some context $\G$ such
  that $\bs \prec_\G \bs'$.
\end{theorem}
\begin{proof}{}
  The  reverse implication, 
  is a direct application of the soundness of the variance judgment.

  The direct implication is proved by induction on
  $\app{\tau}{\ba}$. The variable case is direct: if
  $\app{\alpha}{\sigma} \prec_v \app{\alpha}{\sigma'}$ holds then for
  $\G$ equal to $(v\alpha)$ we indeed have $v\alpha \der \alpha : v$
  and $\sigma \prec_\G \sigma'$.

  In the $\ty{\bt}{t}$ case, we have that
  $\app{(\ty{\bt}{t})}{\bs} \prec_v \app{(\ty{\bt}{t})}{\bs'}$. Suppose
  the variance of $\ty{\ba}{t}$ is $\bar{w\alpha}$: by inversion on
  the head type constructor $\tyc{t}$ we deduce that for each $i$,
  $\app{\tau_i}{\bs} \prec_{\varcomp v {w_i}} \app{\tau_i}{\bs'}$. Our
  induction hypothesis then gives us a family of contexts
  $(\G_i)_\iI$ such that for each $i$ we have
  $\G_i \der \tau_i : \varcomp v {w_i}$. Furthermore,
  $\bs \prec_{\G_i} \bs'$ holds for all $\G_i$, which means
  that $\bs \prec_{\mathop{\wedge}_\iI \G_i} bs'$. Let's define
  $\G$ as $\mathop{\wedge}_\iI \G_i$. By construction we have
  $\G \geq \G_i$, so by monotonicity (Lemma~\ref{lem/monotonicity}) we
  have $\G \der \tau_i : \varcomp v {w_i}$ for each $i$. This
  allows us to conclude $\G \der \ty{\bt}{t} : v$ as desired.
\qed
\end{proof}

Note that this would work even for type constructors that are not
$v$-closed: we are not comparing a $\app{\tau}{\bs}$ to any type
$\tau'$, but to a type $\app{\tau}{\bs'}$ sharing the same
structure---the head constructors are always the same.

For any given pair $\bs, \bs'$ such that
$\app{\tau}{\bs} \prec_v \app{\tau}{\bs'}$ we can produce a context
$\G$ such that $\bs \prec_\G \bs'$. But is there a common
context that would work for any pair? Indeed, that is the lowest
possible context, the principal context $\G$ such that
$\G \der \tau : v$.

\begin{corollary}[Principal inversion] 
\label{cor/principal-inversion}
  If $\Delta$ is minimal for $\Delta \der \tau : v$, then for any
  type sequences $\bs$ and $\bs'$, the subtyping relation
  $\app{\tau}{\bs} \prec_v \app{\tau}{\bs'}$ implies
  $\bs \prec_\Delta \bs'$.
\end{corollary}
\begin{proof}{}
  Let $\Delta$ be the minimal context such that
  $\Delta \der \tau : v$ holds. For any $\bs, \bs'$ such that
  $\app{\tau}{\bs} \prec_v \app{\tau}{\bs'}$, by inversion
  (Theorem~\ref{thm/inversion}) we have some $\G$ such that
  $\G \der \tau : v$ and $\bs \prec_\G \bs'$. By definition of $\Delta$,
  $\Delta \leq \G$ so $\bs \prec_\Delta \bs'$ also holds. That is,
  $\bs \prec_\Delta \bs'$ holds for any $\bs, \bs'$ such that
  $\app{\tau}{\bs} \prec_v \app{\tau}{\bs'}$.
\qed
\end{proof}
\end{version}

\begin{version}{\Esop}
  We can generalize inversion of head type
  constructors~(\S\ref{subtyping_inversion}) to whole type
  expressions. The most general inversion is given by the principal
  context.

\begin{theorem}[Inversion]
\label{thm/inversion}
\label{cor/principal-inversion}
  For any type $\app{\tau}{\ba}$, variance $v$, and type sequences
  $\bs$ and $\bs'$, the subtyping relation
  $\app{\tau}{\bs} \prec_v \app{\tau}{\bs'}$ holds if and only if the
  judgment $\G \der \tau : v$ holds for some context $\G$ such
  that $\bs \prec_\G \bs'$.
  Furthermore, if $\app{\tau}{\bs} \prec_v \app{\tau}{\bs'}$ holds, then
  $\bs \prec_\Delta \bs'$ holds, where $\Delta$ is the minimal context
  such that $\Delta \der \tau : v$. 
\end{theorem}
\end{version}

\subsection{Variance annotations in \ADTs}

As a preparation for the difficult case of  \GADTs,  we first present our
approach in the well-understood case of algebraic datatypes. 
We exhibit a semantic criterion that justifies the
correctness of a variance annotation; then, we propose an equivalent
syntactic judgment. Of course, we recover the usual criterion that
covariant variables should only occur positively. 

In general, an \ADT definition of the form
$$
\type \Gva\t = \big|_{\cC}\; \Kc \of \app{\tau^c}\ba
$$
cannot be accepted with any variance $\ty{\bar{v\alpha}}{t}$.
For example, the declaration $(\type {v\alpha} \mc {inv} = \texttt{Fun of }
\alpha \to \alpha)$ is only sound when $v$ is invariant. Accepting
a variance assignment $\bar{v\alpha}$ determines 
the relations between closed types
$\bs$ and $\bs'$ under which the relation $\ty{\bs}{t} \leq \ty{\bs'}{t}$
is  correct.

In the definition of $\ty{\vplus\alpha}{exp}$ we justified the covariance of
\code{exp}  by the existence of a coercion function. We now
formalize this idea for the general case. To check the correctness of
$\ty{\bs}\t \leq \ty{\bs'}\t$ we check the existence of a coercion term
that turns a closed value $\val$ of type $\ty{\bs}{t}$ into one of type
$\ty{\bs'}{t}$ that is equal to $\val$ up to type information.  
We actually search for coercions of the form:
$$
\match {(\val : \ty{\bs}{\t})}{\branch[|_{c \in C}] 
        {\Kc (x : \app{\tau^c}{\bs})}
        {\Kc (x :> \app{\tau^c}{\bs'})}
}
$$
Note that erasing types gives an $\eta$-expansion of the sum type, 
\ie. this is really a coercion.  Hence, such a coercion
exists if and only if it is well-typed, that is, each cast of the form
$(x : \app{\tau^c}{\bs} :> \app{\tau^c}{\bs'})$ is
itself well-typed. This gives our semantic criterion for \ADTs.

\begin{definition}[Semantic soundness criterion for \ADTs]
\label{def/check_adt/sem}\indent\\
We accept the \ADT definition of 
$\Gva\t$ with constructors
$(\Kc \of \app{\tau^c}{\ba})_{\cC}$ if  
\[ \forall \cC, \forall \bs, \forall \bs', \quad
   \ty{\bs}{t} \leq \ty{\bs'}{t} \implies 
     \app{\tau^c}{\bs} \leq \app{\tau^c}{\bs'}
\]  
\end{definition}

\paragraph{The syntactic criterion for \ADTs} We notice that this
criterion is exactly the semantic interpretation of the variance checking
judgment~(Definition~\ref{def/vc/semantics}): the type $\type \Gva \t$ is
accepted if 
and only if the judgment $\Gva \der \tau^c : \rel\vplus$ is derivable for each
constructor type $\app{\tau^c}{\ba}$. 

This syntactic criterion coincides with the well-known alogrithm
implemented in type checkers\footnote{One should keep in mind that
  this criterion suffers the usual bane of static typing, it can
  reject programs that do not go wrong:
  $\type {\vminus\alpha}\mathtt{weird} = \K \of \alpha * \bot$. For
  more details, see the beginning of the \S\ref{sec/checking_variance_of_gadts}\Esop{ in the
    long version of this
    article}{\S\ref{simonet-pottier-not-complete}}.}: checking positive
occurences of a variable $\alpha$ corresponds to a proof obligation of
the form $\Gva \der \alpha : \vplus$, which is valid only when
$\alpha$ has variance $\rel\vplus$ or $\rel\veq$ in $\Gamma$; checking
negative occurences correspond to a proof obligation
$\Gva \der \alpha : \vminus$, \etc. This extends seamlessly to
irrelevant variables, which must appear only under irrelevant context
$\Gva \der \alpha : \virr$---or not at all.

\subsection{Variance annotations in \GADTs}

\paragraph {A general description of \GADTs}

When used to build terms of type $\ty{\ba}{t}$, a constructor
$\mc{K}~\mc{of}~\tau$ behaves like a function of type $\forall
\ba. (\tau \to \ty{\ba}{t})$. Notice that the codomain is exactly
$\ty{\ba}{t}$, the type $\tyc{t}$ instantiated with parametric
variables. \GADTs arise by relaxing this restriction, allowing
constructors with richer types of the form $\forall \ba. (\tau \to
\ty{\bs}{t})$. See for example the declaration of constructor
\code{Prod} in the introduction:
\begin{lstlisting}[xleftmargin=4em]
  | Prod : $\forall \beta \gamma .\ \ty{\beta}{exp} * \ty{\gamma}{exp}
                                  \to \ty{(\beta*\gamma)}{exp}$
\end{lstlisting}
Instead of being just $\ty{\alpha}{exp}$, the codomain is now
$\ty{(\beta*\gamma)}{exp}$. We moved from simple algebraic datatypes
to so-called \emph{generalized} algebraic datatypes. This approach is
natural and convenient for the users, so it is exactly the syntax
chosen in languages with explicit \GADTs support, such as Haskell and
OCaml, and is reminiscent of the inductive datatype definitions of
dependently typed languages.

\begin{version}{\Not\Esop}
However, for the formal study of \GADTs, a different formulation based on
equality constraints is preferred. The idea is that we will force
again the codomain to be exactly $\ty{\alpha}{exp}$, but allow
additional type equations such as
$\alpha = \beta * \gamma$ in this example:
\begin{lstlisting}[xleftmargin=4em]
| Prod : $\forall \alpha.\ \forall \beta \gamma [ \alpha = \beta * \gamma ].\ 
            \ty{\beta}{exp} * \ty{\gamma}{exp}
            \to \ty{\alpha}{exp}$
\end{lstlisting}
This restricted form justifies the name of \emph{guarded} algebraic
datatype. The $\forall\beta\gamma[D].\tau$ notation, coming from Simonet
and Pottier, is a \emph{constrained type scheme}: $\beta,\gamma$ may
only be instantiated with type parameters respecting the constraint
$D$. Note that, as $\beta$ and $\gamma$ do not appear in the codomain
anymore, we may equivalently change the outer universal into an
existential on the left-hand side of the arrow:
\begin{lstlisting}[xleftmargin=4em]
| Prod : $\forall \alpha.\ (\exists \beta \gamma [ \alpha = \beta * \gamma ].\ 
            \ty{\beta}{exp} * \ty{\gamma}{exp})
            \to \ty{\alpha}{exp}$
\end{lstlisting}
In the general case, a \GADT definition for $\ty{\ba}{t}$ is composed of a set
of constructor declarations, each of the form:
\begin{lstlisting}[xleftmargin=4em]
| K : $\forall \ba.\
         (\exists \bb [ \ba = \app \bs \bb ] .\,\app \tau {\ba,\bb})
         \to \ty{\ba}{t}$
\end{lstlisting}
or, reusing the classic notation,
\begin{lstlisting}[xleftmargin=4em]
| K of $\exists \bb [ \ba = \app \bs \bb ] .\,
                            \app \tau {\ba,\bb}$
\end{lstlisting}
Without loss of generality, we can conveniently assume that the
variables $\ba$ do not appear in the parameter type $\tau$
anymore: if some $\alpha_i$ appears in $\tau$, one may always pick
a fresh existential variable $\beta$, add the constraint
$\alpha = \beta$ to $D$, and consider $\tau[\beta/\alpha]$. Let us
re-express the introductory example in this form, that is, 
$\mc{K of}~\exists\bb[\ba = \app \bs \bb].\, \app \tau \bb$:
\end{version}

\begin{version}{\Esop}
  However, for the formal study of \GADTs, a different formulation
  based on equality constraints is preferred. We use the
  following equivalent presentation, already present in previous
  works~\cite{simonet-pottier-hmg-toplas}. We force the codomain
  of the constructor \code{Prod}  to be  $\ty{\alpha}{t}$ again, instead of
  $\ty{(\beta * \gamma)}{t}$,  by adding an explicit equality constraint
  $\alpha = \beta * \gamma$. 
\end{version}

\begin{lstlisting}[xleftmargin=4em]
type $\ty{\alpha}{exp}$ =
  | Val of $\exists \beta [\alpha = \beta].\,\beta$
  | Int of $[\alpha = \tyc{int}].\,\tyc{int}$
  | Thunk of $\exists \beta \gamma [\alpha = \gamma].\, \ty{\beta}{exp} * (\beta \to \gamma)$
  | Prod of $\exists \beta\gamma[\alpha = \beta*\gamma].\, \ty{\beta}{exp} * \ty{\gamma}{exp}$
\end{lstlisting}

\begin{version}{\Not\Esop}
If all constraints between brackets are of the simple form
$\alpha_i = \beta_i$ (for distinct variables $\alpha_i$ and $\beta_i$), as for
the constructor \code{Thunk}, then we have a constructor with
existential types as described by Laüfer and Odersky
\cite{laufer-odersky-92}. If furthermore there are no other
existential variables than those equated with a type parameter, as in
the \code{Val} case, we have an usual algebraic type constructor; of
course the whole type is ``simply algebraic'' only if each of its
constructors is algebraic.
\end{version}

In the rest of the paper, we extend our former core language with such
definitions. This does not impact the notion of
subtyping, which is defined on \GADT type constructors with variance
$\mc{type}~\ty{\Gva}{t}$ just as it previously was on simple \ADT type
constructors. What needs to be changed, however, is the soundness
criterion for checking the variance of type definitions

\paragraph{The correctness criterion}

We must adapt our semantic criterion for datatype
declarations~(Definition~\ref{def/check_adt/sem}) from simple \ADTs to
\GADTs. 
Again, 
we check under which relations between $\bs$ and $\bs'$
the subtyping relation $\ty{\bs}{t} \leq
\ty{\bs'}{t}$ holds for some \GADT definition $\ty{\Gva}{t}$.

The difference is that a constructor \texttt{K$_c$}
that had an argument of type $\app{\tau^c}{\ba}$ in the simple \ADT
case, now has the more complex type $\exists \bb
[\app{D}{\ba,\bb}]. \app{\tau^c}{\bb}$, for a set of existential
variables $\bb$ and a set of equality constraints $D$---of the form
$(\alpha_i = \app{T_i}{\bb})_\iI$ for a family of type expressions
$(\app{T_i}{\bb})_\iI$.  
Given a closed value $\val$ of type $\ty\bs\t$, the coercion term is:
$$
\match {(\val : \ty{\bs}{\t})}{\branch[|_{c \in C}] 
        {\Kc (x : \app{\tau^c}{\br_c})}
        {\Kc (x :> \app{\tau^c}{\br_c'})}
}
$$
We do not need to consider the dead cases: we only match on the 
constructors for which there exists an instantiation $\br_c$ of the
existential variables $\bb$ such that the constraint $\app{D}{\bs,\br}$,
\ie. $\bigwedge_\iI \sigma_i = \app{T_i}{\br_c}$, holds. To
type-check this term, we need to find another instantiation $\br'_c$
that verifies the constraints $\app{D}{\bs',\br'}$. This coercion
type-checks only when $\app{\tau^c}{\br_c} \leq \app{\tau^c}{\br'_c}$
holds. This gives our semantic criterion for \GADTs:

\begin{definition}[Semantic soundness criterion for \GADTs]
\label{def/check_gadt/sem}
We accept the \GADT definition of $\mc{type}~\ty{\Gva}{t}$
with constructors $(\Kc \of \exists \bb
[\app{D}{\ba,\bb}]. \app{\tau^c}{\ba})_\cC$, if for all $c$ in $C$ we have:
$$
\forall \bs, \bs', \br, \uad
\bigparens{
   \ty{\bs}{t} \leq \ty{\bs'}{t}
     \wedge \app D{\bs,\br}
   \implies
   \exists \br',\; \app D{\bs',\br'}
     \wedge \app\tau\br \leq \app\tau{\br'}
}
\eqno 
(\DefRule {Req})
$$
\end{definition}
As for \ADTs, this criterion ensures soundness: if, under some
variance annotation, a datatype declaration satisfies it, then the
implied subtyping relations are all expressible as coercions in the
language, and therefore correct. Whereas the simpler \ADT criterion
was already widely present in the literature, this one is less
known; it is however present in the previous work of Simonet and
Pottier~\cite{simonet-pottier-hmg-toplas} (presented as a constraint
entailment problem). 

Another way to understand this criterion would be to define
constrained existential types of the form $\exists \bb
[\app{D}{\ba,\bb}].\app{\tau}{\bb}$ as first-class types and, with the
right notion of subtyping for those, require that $\ty{\bs}{t} \leq
\ty{\bs'}{t}$ imply $(\exists \bb [\app{D}{\bs,\bb}].\app{\tau}{\bb})
\leq (\exists \bb [\app{D}{\bs',\bb}].\app{\tau}{\bb})$. The (easy)
equivalence between those two presentations is detailed in the work of
Simonet and Pottier~\cite{simonet-pottier-hmg-toplas}. 

\section{Checking variances of \GADT}
\label{sec/checking_variance_of_gadts}

\begin{version}{\Not\Esop}
\paragraph{A remark on the non-completeness}
\label{simonet-pottier-not-complete}

Note that while the criterion \Rule {Req} is \emph{sound}, it is
not \emph{complete}---even in the simple \ADT case. Consider indeed
the following example, using an uninhabited type $\bot$:
\begin{lstlisting}
type $\ty{+\alpha}{t}$ =
  | T of $\tyc{int}$
  | Empty of $\bot * (\alpha \to \tyc{bool})$
\end{lstlisting}
Despite $\alpha$ occurring in a contravariant position in the
\code{Empty} constructor, we know that we will never encounter a value
with this constructor as its argument type in uninhabited. The
contravariant occurrence will therefore never make a program ``go
wrong''. The definition is correct, yet rejected by \Rule {Req} (or
the classic \ADT criterion), which is therefore ``incomplete''.

In term of coercions, the semantic criterion fails because the
\texttt{Empty} clause cannot be type-checked. We would need the type
system to realize that, because of the uninhabited argument type, it
is a dead branch anyway. Deciding type inhabitation in the general
case is a very complex question, which is mostly orthogonal to the
presence and design of \GADTs in the type system. There is, however,
one clear interaction between the type inhabitation question and
\GADTs. If a \GADT $\ty{\ba}{t}$ is instantiated with type variables
$\bs$ that satisfy none of the constraints $\app{D}{\alpha}$ of its
constructors $\mc{K}~\mc{of}~\exists\bb[\bar{D}].\tau$, then we know
that $\ty{\bs}{t}$ is not inhabited. This is related to the idea of
``domain information'' that we discuss in the Future Work
section~(\S\ref{future_work:domain_information}).

In the rest of the article, the discussion on \emph{completeness}
should be understood with respect to the fixed semantic criterion: we
look for a syntactic check that is as good as the criterion -- which
is accepted as the correct check in the existing literature.
\end{version}

\subsection{Expressing decomposability}
\label{sec/expressing_decomposability}

If we specialize \Rule {Req} to the \code{Prod} constructor
of the $\ty{\alpha}{exp}$ example datatype, \ie. 
$\mc{Prod}~\mc{of}~\exists \beta \gamma[\alpha = \beta * \gamma]
\ty{\beta}{exp} * \ty{\gamma}{exp}$, we get:
\[
\begin{array}{l}
  \forall \sigma, \sigma', \rho_1, \rho_2, \\\uad
  \bigparens{
  \ty{\sigma}{exp} \leq \ty{\sigma'}{exp}
    \wedge
    \sigma = \rho_1 * \rho_2
  \implies
  \exists \rho'_1, \rho'_2, \parens {\sigma' = \rho'_1 * \rho'_2
    \wedge
    \rho_1 * \rho_2 \leq \rho'_1 * \rho'_2}
  }
\end{array}
\]
We can substitute equalities and use the (user-defined) covariance to
simplify the subtyping constraint
$\ty{\sigma}{exp} \leq \ty{\sigma'}{exp}$ into $\sigma \leq
\sigma'$:
\[
  \forall \sigma', \rho_1, \rho_2,\uad
  \bigparens
    {\rho_1 * \rho_2 \leq \sigma'
  \implies
  \exists \rho'_1, \rho'_2, \uad  
    \parens 
      {\sigma' = \rho'_1 * \rho'_2
       \wide\wedge
       \rho_1 \leq \rho'_1
       \wide\wedge
       \rho_2 \leq \rho'_2
      }
  }
\eqno \llabel 1
\]
This is the \emph{upward closure} property mentioned in the
introduction. 
The preceeding transformation is safe only if any supertype
$\sigma'$ of a product $\rho_1 * \rho_2$ is itself a product, \ie. is
of the form $\rho'_1 * \rho'_2$ for some $\rho'_1$ and $\rho'_2$.

More generally, for a type $\G \der \sigma$ and a variance $v$, we are
interested in a closure property of the following form, where the notation
$(\br : \G)$ simply classifies type vectors $\br$ that have exactly
one type $\rho_i$ for each variable in $\Gamma$: \[
  \forall (\br : \G), \sigma',\quad
    \app \sigma \br \prec_{v} \sigma' \implies
    \exists (\br' : \G),\,\sigma' = \app \sigma {\br'}
\]
Here, the context $\G$ represents the set of existential variables
of the constructor ($\beta$ and $\gamma$ in our example). We can
easily express the condition $\rho_1 \leq \rho'_1$ and
$\rho_2 \leq \rho'_2$ on the right-hand side of the implication by
considering a context $\G$ annotated with variances
$(\vplus\beta, \vplus\gamma)$, and using the context ordering
$\rel{\prec_\G}$. Then, \lref 1 is equivalent to: 
$$
  \forall (\br : \G), \sigma',\quad
    \app \sigma \br \prec_{v} \sigma' \implies
    \exists (\br' : \G),\uad\br \prec_\G \br'
      \wedge \sigma' = \app \sigma {\br'}
$$
Our aim is now to find a set of inference rules to check decomposability; we
will later reconnect it  to \Rule {Req}.
In fact, we study a slightly more general relation, where the
equality $\app\sigma{\br'}~=~\sigma'$ on the right-hand side is
relaxed to an arbitrary relation
$\app{\sigma}{\br'} \prec_{v'} \sigma'$:

\begin{definition}[Decomposability]
\label{def/decomposability}
Given a context $\G$, a type expression $\app{\sigma}{\bb}$ and two
variances $v$ and $v'$, we say that $\sigma$ is \emph{decomposable} under
$\G$ from variance $v$ to variance $v'$, which we write $\G \Der \sigma : v
\vto v'$, if the following property holds: $$
\forall (\br : \G), \sigma',\uad
      \app \sigma \br \prec_{v} \sigma' \implies
      \exists (\br' : \G),\uad
        \br \prec_\G \br'
        \wide\wedge \app \sigma {\br'} \prec_{v'} \sigma'
$$
\end{definition}
We use the symbol $\Der$ rather than $\der$ to highlight the fact that
this is just a logic formula, not the semantic interpretation of
a syntactic judgment---we will introduce one later in
section~\ref{sec/syntactic-decomposability}.

Remark that, due to the \emph{positive} occurrence of the relation
$\prec_\G$ in the proposition ${\G \Der \tau : v \vto v'}$ and the
anti-monotonicity of $\prec_\G$, this formula is ``anti-monotone''
with respect to the context ordering $\G \leq \G'$. This corresponds
to saying that we can still decompose, but with less information on
the existential witness $\br'$.

\begin{lemma}[Anti-monotonicity]
\label{lem/anti-monotonicity}\indent\\
If $\G \Der \tau : v \vto v'$ holds
and $\G' \leq \G$, then $\G' \Der \tau : v \vto v'$ also holds.
\end{lemma}
\begin{version}{\Not\Esop}
Our final decomposability criterion, given below in Figure~\ref{fig/decomposability},
requires both correct variances and a decomposability property, so it
will be neither monotone nor anti-monotone with respect to the
context argument.
\end{version}

\subsection{Variable occurrences}

In the \code{Prod} case, the type whose decomposability was considered
is $\beta * \gamma$ (in the context $\beta, \gamma$).  In this very
simple case, decomposability depends only on the type constructor for
the product. In the present type system, with very strong
invertibility principles on the subtyping relation, both upward and
downward closures hold for products. In the general case, we require that
this specific type constructor be upward-closed.

In general, the closure of the head type constructor alone is
not enough to ensure decomposability of the whole type. For example,
in a complex type expression with subterms, we should consider the
closure of the type constructors appearing in the subterms as
well. Besides, there are subtleties when a variable occurs several
times.

For example, while $\beta * \gamma$ is decomposable from $\rel\vplus$
to $\rel\veq$, $\beta * \beta$ is not: $\bot * \bot$ is an
instantiation of $\beta * \beta$, and a subtype of, \eg., $\tyc{int} *
\tyc{bool}$, which is not an instance\footnote {We use the
term \emph{instance} to denote  the replacement  of all the free variables
of a type expression under context by   closed types---not the
specialization of an ML type scheme.} of $\beta * \beta$. The same variable 
occurring twice in covariant 
position (or having one covariant and one invariant or contravariant
occurence) breaks decomposability.

On the other hand, two invariant occurrences are possible:
$\ty{\beta}{ref} * \ty{\beta}{ref}$ is upward-closed (assuming the
type constructor $\tyc{ref}$ is invariant and upward-closed): if
$(\ty{\sigma}{ref} * \ty{\sigma}{ref}) \leq \sigma'$, then by
upward closure of the product, $\sigma'$ is of the form
$\sigma'_1 * \sigma'_2$, and by its covariance
$\ty{\sigma}{ref} \leq \sigma'_1$ and
$\ty{\sigma}{ref} \leq \sigma'_2$. Now by invariance of $\tyc{ref}$ we
have $\sigma'_1 = \ty{\sigma}{ref} = \sigma'_2$, and therefore
$\sigma'$ is equal to $\ty{\sigma}{ref} * \ty{\sigma}{ref}$, which is
an instance of $\ty{\beta}{ref} * \ty{\beta}{ref}$. 

Finally, a variable may appear in irrelevant positions without
affecting closure properties; $\beta * (\ty{\beta}{irr})$
(where $\tyc{irr}$ is an upward-closed irrelevant type, defined for
example as $\mc{type}~\ty{\alpha}{irr} = \tyc{int}$) is upward closed:
if $\sigma * (\ty{\sigma}{irr}) \leq \sigma'$, then $\sigma'$ is of
the form $\sigma'_1 * (\ty{\sigma'_2}{irr})$ with
$\sigma \leq \sigma'_1$ and $\sigma \Join \sigma'_2$, which is
equiconvertible to $\sigma'_1 * (\ty{\sigma'_1}{irr})$ by irrelevance,
an instance of $\beta * (\ty{\beta}{irr})$.

\subsection{Context zipping}

The intuition to think about these different cases is to consider
that, for any~$\sigma'$, we are looking for a way to construct a ``witness''
$\bs'$ such that ${\app \tau {\bs'} = \sigma'}$ 
from the hypothesis $\app \tau \bs \prec_v \sigma'$.
When a type variable appears only once, its witness can be determined
by inspecting the corresponding position in the type~$\sigma'$. For
example, in $\alpha * \beta \leq \tbool * \tint$, the mapping $\alpha
\mapsto \tbool, \beta \mapsto \tint$ gives the witness pair $\tbool,
\tint$.

However, when a variable appears twice, the two witnesses
corresponding to the two occurrences may not coincide. (Consider for
example $\beta * \beta \leq \tyc{bool} * \tyc{int}$.) If a variable
$\beta_i$ appears in several \emph{invariant} occurrences, the witness
of each occurrence is forced to be equal to the corresponding subterm
of $\app \tau \bs$, that is $\sigma_i$, and therefore the various
witnesses are themselves equal, hence compatible.  On the contrary,
for two covariant occurrences (as in the $\beta * \beta$ case), it is
possible to pick a $\sigma'$ such that the two witnesses are
incompatible---and similarly for one covariant and one invariant
occurrence. Finally, an irrelevant occurrence will never break closure
properties, as all witnesses (forced by another occurrence) are
compatible.

To express these merging properties, we define a \emph{zip}
operation $v_1 \zip v_2$, that
formally expresses which combinations of variances are possible for
several occurrences of the same variable; it is a partial operation
(for example, it is not defined in the covariant-covariant case, which
breaks the closure properties) with the following table:
$$
\def \C#1{\multicolumn{1}{C}{#1}}
\begin{tabular}{R|c@{}|C|C|C|C|l}
  v \zip w && \C{\veq} & \C{\vplus} & \C{\vminus} & \C{\virr} & $w$
\\ 
\cline{1-1}\cline{3-7}\noalign{\vskip \doublerulesep}\cline{1-1}\cline{3-6}
  \veq     && \veq   &          &           & \veq    \\ \cline{3-6}
  \vplus   &&        &          &           & \vplus  \\ \cline{3-6}
  \vminus  &&        &          &           & \vminus \\ \cline{3-6}
  \virr    && \veq   & \vplus   & \vminus   & \virr   \\ \cline{3-6}
  v        &\multicolumn{5}{C}{}
\end{tabular}
$$

\begin{version}{\Not\Esop}
The following lemma uses zipping to merge together the results of the
decomposition of several subterms $(\sigma_i)_i$ into a ``simultaneous
decomposition''.

\begin{definition}[Simultaneous decomposition]
\label{def/simultaneous-decomposition}
Given a context $\G$, and families of type expressions $(\sigma_i)_\iI$ and
variances $(v_i)_\iI$ and $(v'_i)_\iI$, we define the following
``simultaneous closure property'' ${\G \Der (T_i : v_i \vto v'_i)_\iI}$
defined as :
$$
    \forall (\br : \G), \bs',\uad
      \parens 
        {\forall \iI, \app{\sigma_i}{\br} \prec_{v_i} \sigma'_i}
      \implies
      \exists (\br' : \G),\uad \br \prec_\G \br'
        \wedge \parens 
         {\forall \iI, \app {\sigma_i} {\br'} \prec_{v'_i} \sigma'_i}
\nobelowdisplayskip
$$
\end{definition}
\begin{lemma}[Soundness of zipping]
\label{lem/zip-soundness}
Suppose we have families of type expressions $(\app{T_i}{\bb})_\iI$, contexts
$(\G_i)_\iI$ and variances $(v_i)_\iI$ and $(v'_i)_\iI$ such that
$\mathop{\zip}_\iI \G_i$ exists and for all $i$ we have both
$\G_i \der T_i : v_i$ and 
$\G_i \Der T_i : v_i \vto v'_i$.
Then, we have $(\mathop{\zip}_\iI \G_i) \Der (\sigma_i : v_i \vto v'_i)_\iI$.
\end{lemma}

\begin{proof}{}
\locallabelreset
Without loss of generality, 
\XXX[TODO]{Say why...}
we can consider that there are only two
type expressions $\app{T_1}{\bb}$ and $\app{T_2}{\bb}$, and that the
free variables $\bb$ is reduced to  a single variable $\beta$. 
Let $w_1, w_2$ be the respective variances of $\beta$ in $\G_1,
\G_2$. 
We know that $(w_1 \zip w_2)$ exists  and is equal to the variance $w$
of the variable $\beta$ in $\G_1 \zip \G_2$. 

\let \br \rho

Our further assumptions are 
$\G_i \der T_i : v_i$~\llabel 1 and 
$\G_i \Der T_i : v_i \vto v'_i$~\llabel 2 
for $i$ in $\set {1,2}$. 
The  expansion of \lref 2 is:
$$
  \forall \iI,\, \forall \br,\sigma'_i,\quad
     \app{T_i}{\br} \prec_{v_i} \sigma'_i
     \wide\implies \exists \br', \uad \br \prec_{\G_i} \br'
       \;\wedge\;\app{T_i}{\br'} \prec_{v'_i} \sigma'_i
\eqno \llabel 3
$$
Our goal is to prove  $\G_1 \zip \G_2 \Der (\sigma_i : v_i\vto
v'_i)_\iI$, which is equivalent to:  
$$
  \forall \br,\bs',\quad
    (\forall \iI, (\app{T_i}{\br} \prec_{v_i} \sigma'_i))
    \wide\implies \exists \br',\uad
        \br \prec_{\G} \br' \wide\wedge
        \forall \iI, \app{T_i}{\br'} \prec_{v'_i} \sigma'_i
\eqno \llabel C
$$
Assume given $\rho$ and $(\sigma'_1, \sigma'_2)$ such that
$\app{T_1}{\br} \prec_{v_1} \sigma'_1$~\llabel {H1} and
$\app{T_2}{\br} \prec_{v_2} \sigma'_2$~\llabel {H2}. 
Applying~\lref 3 with $i$ equal to $1$ and \lref {H1} ensures the existence of 
a $\rho'_1$ such that $\rho \prec_{w_1} \rho'_1$~\llabel {rw1} and
$\app{T_1}{\rho'_1} \prec_{v_1} \sigma'_1$~\llabel {rp1}. 
Similarly, there exists  $\rho'_2$ such that
$\rho \prec_{w_2} \rho'_2$~\llabel {rw2} and
$\app{T_2}{\rho'_2} \prec_{v_2} \sigma'_2$~\llabel {rp2}. 
To establish \lref C, it remains to build a single $\rho'$ that satisfies
$\rho \prec_w \rho'$~\llabel {r0}, $\app{T_1}{\rho'} \prec_{v_1}
\sigma'_1$~\llabel {r1} and $\app{T_2}{\rho'} \prec_{v_2} \sigma'_2$~\llabel
{r2}, simultaneously.
We reason by case analysis on $w_1$ and $w_2$ (restricted to the cases where
the zip exists).

If both $w_1$ and $w_2$ are $\rel\virr$, we take either $\rho'_1$ or
$\rho'_2$ for: since $w$ is $\virr$, we \lref {r0} trivially holds; 
Furthermore, $\app{T}{\rho'} = \app{T}{\rho'_1} = \app{T}{\rho'_2}$ by 
irrelevance of $v_i$ and \lref 1; therefore \lref {r1} and \lref {r2} 
follow from \lref {rp1} and \lref {rp2}.

If only one of the $w_i$ is $\rel\virr$, we'll suppose that it is $w_1$. We
then take $\rho_2$ for $\rho'$.  Since $\virr_1 \zip w_2$ is $w_2$, \lref
{r0} follows from \lref {rw2}; Furthermore, $\app{T}{\rho'} =
\app{T}{\rho'_1}$ by irrelevance of $v_1$ and \lref 1 while
$\app{T_2}{\rho'} = \app{T_2}{\rho'_2}$ holds by construction; Hence, as in
the previous case, \lref {r1} and \lref {r2} follow from \lref {rp1} and
\lref {rp2}.

Finally, if both $w_1$ and $w_2$ are $\rel\veq$, then 
\lref {rw1} and \lref {rw2} implies $\rho'_1 = \rho = \rho'_2$.
We take $\rho$ for $\rho'$ and all three conditions are obviously satisfied.
\qed
\end{proof}
This lemma admits a kind of converse lemma stating completeness of
zipping, that says that if ${\G \Der (T_i : v_i \vto v'_i)_i}$ holds, then
$\G$ is indeed related to a zip of contexts $(\G_i)_i$ that pairwise
decompose each of the $(T_i)_i$. However, the proof of completeness is more
delicate and
\Esop{can only be found in the extended version}
{we prove it separately in \S\ref{sec/zip-completeness}}.
\end{version}

\subsection{Syntactic decomposability}

\label {sec/syntactic-decomposability}

\begin{mathparfig}{fig/decomposability}{Syntactic decomposablity}
\inferrule[sc-Triv]
  {v \geq v'\\\G \der \tau : v}
  {\G \der \tau : v \To v'}

\inferrule[sc-Var]
  {w\alpha \in \G \\ w = v}
  {\G \der \alpha : v \To v'}

\inferrule[sc-Constr]
  {\G \der \mc{type}~\ty{\bar{w\alpha}}{t} : v\textrm{-closed}\\
   \G = \mathop{\zip_i} \G_i\\
   \forall i,\ \G_i \der \sigma_i : \varcomp v {w_i} \To \varcomp{v'}{w_i}}
  {\G \der \ty{\bar{\sigma}}{t} : v \To v'}
\end{mathparfig}

Equipped with the zipping operation, we introduce a judgment
$\G \der \tau : v \To v'$ to express decomposability, syntactically,
defined by the inference rules on Figure~\ref{fig/decomposability}. We
also define its semantic interpretation
$\sem{\G \der \tau : v \To v'}$. The judgment and its interpretation
were co-designed, so keeping the interpretation in mind is
the best way to understand the subtleties of the inference rules. We
\Esop
{use zipping, which requires correct variances, to}
{rely on the zip soundness (Lemma~\ref{lem/zip-soundness}) to}
merge sub-derivations into larger ones, so, in addition to
decomposability, the interpretation also ensures that $v$ is a correct
variance for $\tau$ under $\G$. This subtlety is why we have two
different properties for decomposability, $\G \Der \tau : v \vto v'$
and $\sem{\G \der \tau : v \To v'}$.

\begin{definition}[Interpretation of syntactic decomposability]
\label{def/cc/semantics}\noindent\\
We write $\sem{\G \der \tau : v \To v'}$ for the conjunction of
properties $\sem{\G \der \tau : v}$ and $\G \Der \tau : v \vto v'$.
\end{definition}

To understand the inference rules, the first thing to notice is that
the present rules are not completely syntax-directed: we first check
whether $v \geq v'$ holds, and if not, we apply syntax-directed
inference rules; existence of derivations is still easily
decidable. If $v \geq v'$ holds, satisfying
$\G \Der \tau : v \vto v'$~(Definition~\ref{def/decomposability}) is
trivial: $\app \tau \bs \prec_v \tau'$ implies
$\app \tau \bs \prec_{v'} \tau'$, so taking $\bs$ for $\bs'$ is always
a correct witness, which is represented by Rule \Rule{sc-Triv}.  The
other rules then follow the same structure as the variance-checking
judgment.

Rule \Rule{sc-Var} is very similar to \Rule {vc-Var}, except that the
condition $w \geq v$ is replaced by a stronger equality $w = v$.
\begin{version}{\Not\Esop}
  The reason why the variance-checking judgment has an inequality
  $w \geq v$ is to make it monotone in the environment---as
  requested by its corresponding semantic
  criterion~(Definition \ref{def/vc/semantics}). Therefore, the
  condition $w \geq v$ is necessary for completeness---and admissible.
  On the contrary, the present judgment ensures, according the
  semantic criterion~(Definition~\ref{def/cc/semantics}), that both
  the variance is correct (monotone in the environment) and the type
  is decomposable, a property which is \emph{anti-monotonic} in the
  environment~(Lemma~\ref{lem/anti-monotonicity}).  Therefore, the
  semantic criterion $\sem{\G \der \tau : v \To v'}$ is invariant
  in $\G$ and, correspondingly, the variable rule must use a strict
  equality.
\end{version}
\begin{version}{\Esop}
  This difference comes from the fact that the semantic condition for
  closure checking~(Definition \ref{def/vc/semantics}) includes both
  a variance check, which is monotonic in the
  context~(Lemma \ref{lem/monotonicity}) and the decomposability property,
  which is anti-monotonic~(Lemma~\ref{lem/anti-monotonicity}), so the
  present judgment must be invariant with respect to the context.
\end{version}

The most interesting rule is \Rule{sc-Constr}. It checks first that
the head type constructor is $v$-closed (according to Definition~\ref
{def/v-closed}); then, it checks that each subtype is decomposable from
$v$ to $v'$, \emph{with compatible witnesses}, that is, in an
environment family $(\G_i)_\iI$ that can be zipped into a unique
environment $\G$.


\begin{lemma}[Soundness of syntactic decomposability]
\label {lem/decomposability-soundness}\noindent\\
If the judgment $\G \der \tau : v \To v'$ holds, then $\sem{\G \der
\tau : v \To v'}$ holds.
\end{lemma}
\begin{proof}{}
\locallabelreset
The proof is by induction on the derivation $\G \der \tau : v \To
v'$~\llabel H.  Expanding $\sem{\G \der \tau : v \To v'}$, we must show both
$\sem{\G \der \tau : v}$, or equivalently $\G \der \tau : v$~\llabel {C1}, 
and $\G \Der \tau : v \vto v'$~\llabel {C2}, which itself expands to:
$$
    \forall (\br : \G), \sigma',\quad
      \app \sigma \br \prec_{v} \sigma' \implies
      \exists (\br' : \G),\uad 
        \br \prec_\G \br'
        \wedge \app \sigma {\br'} \prec_{v'} \sigma'
$$
Let  $\br$, $\tau'$ be such that $\app{\tau}{\br} \prec_v \tau'$. 
We must exhibit a sequence $\br'$ such
that $\br \prec_\G \br'$~\llabel {Cr} and 
$\app{\tau}{\br'} \prec_{v'} \tau'$~\llabel {Ct}.
Cases where the derivation of \lref H ends with \Rule{sc-Triv} and
\textsc{sc-Var} cases are direct: take $\br$ and $(\dots, \sigma', \dots)$
for $\br'$, respectively.

In the remaining cases, the derivation ends with Rule \Rule{sc-Constr} and
$\tau$ is of the form $\ty{\bs}{t}$.  
\begin{itemize}
\item 
The $v$-closure assumption of the left
premise ensures that $\tau'$ is itself of the form $\ty{\bs'}{t}$ for some
sequence of closed types $\bs'$.
By inversion on the variance $\bar{w\alpha}$ of the head constructor
$\tyc{t}$, we deduce $\app{\sigma_i}{\br} \prec_{\varcomp v {w_i}}
\sigma'_i$ for all $i$~\llabel A. 
\item
The middle premise is the zipping assumption on the contexts $\G =
\mathop{\zip}_\iI \G_i$~\llabel C. 
\item
The right premises give us subderivations $\G_i \der
\sigma_i : \varcomp v {w_i} \To \varcomp {v'} {w_i}$ for all $i$.
This implies $\G_i \der \sigma_i : \varcomp v {w_i}$ for all $i$, 
which implies $\G \der \ty{\bs}{t} : v$, \ie.~\lref {C1}.
By induction hypothesis, this also implies 
$\G_i \der \sigma_i : \varcomp v {w_i}$~\llabel D and 
$\G_i \Der \sigma_i : \varcomp v {w_i} \To \varcomp {v'} {w_i}$ for
all $i$~\llabel E. 
\end{itemize}
We may now apply zip soundness (Lemma~\ref{lem/zip-soundness}) with
hypotheses \lref C, \lref D and \lref E, which  gives us the simultaneous
decomposition 
$\G \Der (\sigma'_i : \varcomp v {w_i} \vto \varcomp {v'}
{w_i})_\iI$.   Expanding this property (Definition \ref
{def/simultaneous-decomposition}), we may apply to \lref A to get 
to get a witness $\br'$ such that both $\br' \prec_\G \br$, \ie.
our first goal \lref {Cr},
and $(\forall \iI, \app{\sigma_i}{\br'} \prec_{\varcomp {v'} {w_i}}
\sigma'_i)$, which implies 
$\app{(\ty{\bs}{t})}{\br'} \prec_{v'} \ty{\bs'}{t}$, \ie. our second goal
\lref {Ct}.
\qed
\end{proof}

Completeness is the general case is however much more difficult and we
only prove it when the right-hand side variance $v'$ is $\rel\veq$. In
other words, we take back the generality that we have introduced
in~\S\ref{sec/expressing_decomposability} when defining
decomposability.
\begin{version}{\Not\Esop}
The proof requires several auxiliary lemmas; it is the
subject of the next subsection.
\end{version}

\begin{version}{\Esop}
\begin{lemma}[Completeness of syntactic decomposability]\noindent\\
\label{lem/decomposability-completeness}
If $\sem{\G \der \tau : v \To v'}$ holds for $v' \in \{\veq,\virr\}$,
then $\G \der \tau : v \To v'$ is provable.
\end{lemma}

Lemma~\ref{lem/decomposability-completeness} is an essential 
piece to finally turn the semantic criterion \Rule{Req} into a purely
syntactic form.
\end{version}

\begin{version}{\Not\Esop}

\subsection{Completeness of syntactic decomposability}

\label {sec/zip-completeness}

We first show a few auxiliary results that will serve in the proof of
zip completeness, and later, to reconnect our closure-checking
criterion~(Definition \ref{def/cc/semantics}) with the
full criterion of Simonet and Pottier~(\Rule{Req}).

\begin{lemma}[Intermediate value]\label{intermediate_value_lemma}
  Let  $\app{\tau}{\ba}$ be a type expression and 
  $\br_1$, $\br_2$, $\br_3$ three type families such that
  $\app{\tau}{\br_1} \leq \app{\tau}{\br_2} \leq \app{\tau}{\br_3}$ and 
  $\br_1 \prec_\G \br_3$ holds for some $\G$.
Then, there   exists a type family $\br_2'$ such that both
$\br_1 \prec_\G \br_2' \prec_\G \br_3$ and
$\app{\tau}{\br_2'} = \app{\tau}{\br_2}$ hold.
\end{lemma}

\begin{proof}{}
\locallabelreset
We reuse the notations of the definition and assume
$\app{\tau}{\br_1} \leq \app{\tau}{\br_2} \leq \app{\tau}{\br_3}$~\llabel {T123} and 
$\br_1 \prec_\G \br_3$~\llabel {G13}. 
We just have to exhibit $\rho_2'$ such that both
$\br_1 \prec_\G \br_2' \prec_\G \br_3$\llabel {C123} and
$\app{\tau}{\br_2'} = \app{\tau}{\br_2}$~\llabel {C22'} hold.
Let $\Delta$ be the most general variance of $\tau$, \ie.  the
lowest context such that $\Delta \der \tau : \vplus$~\llabel {Delta} holds. 
By principal inversion~(Corollary \ref{cor/principal-inversion})
applied to \lref {T123} twice thanks to \lref {Delta}, we have
$\br_1 \prec_\Delta \br_2 \prec_\Delta \br_3$~\llabel {br123}.

If $\Delta \geq \G$, the result is immediate, as 
$\br_1 \prec_\G \br_2 \prec_\G \br_3$ follows from
\lref {br123} by anti-monotonicity and both \lref {C123} and \lref {C22'}
hold when we take $\br_2$ for $\br_2'$.
Otherwise, we  reason on each variable of the context $\G$ independently. 
\XXX{If we moved this \wlg. assumption up front, we would have simpler notations}
We may assume, \wlg., that $\tau$ is defined over
a single free variable $\alpha$, and $\G$ and $\Delta$ are
single variances $v_\Delta, v_\G$ with
$v_\Delta \ngeq v_\G$. We reason by case analysis on the
possible variances for $(v_\Delta, v_\G)$, which are
$\{(\any,\veq), (\vplus,\vminus), (\vminus,\vplus), (\virr,\any)\}$.

If $v_\G$ is $\rel\veq$, the hypotheses \lref {G13} 
and  \lref {T123} become $\rho_1 = \rho_3$ and
$\app{\tau}{\rho_1} \leq \app{\tau}{\rho_2} \leq \app{\tau}{\rho_1}$, 
which implies $\app{\tau}{\rho_1} = \app{\tau}{\rho_2}$. Thus, taking
$\rho_1$ for $\rho_2'$ satisfies \lref {C123} and~\lref {C22'}. 

If $v_\Delta$ is $\rel\virr$, then by irrelevant of $\virr$ and \lref
{Delta}, we have $\app{\tau}{\rho_1} = \app{\tau}{\rho_2}$. Thus, taking
$\rho_1$ for $\rho'_2$ satisfies \lref {C123} and~\lref {C22'},  as above. 

Finally, the cases $(\vplus,\vminus)$ and $(\vminus,\vplus)$ are
symmetric and we will only work out the first one, \ie. 
$v_\Delta$ is $\rel\vplus$ and $v_\G$ is $\rel\vminus$. From
$\app{\tau}{\rho_1} \leq \app{\tau}{\rho_3}$ (which follows from \lref {T123}
by transitivity) and \lref {Delta}, we have 
$\rho_1 \prec_{v_\Delta} \rho_3$, \ie. $\rho_1 \leq \rho_3$.
Since the hypothesis~\lref {G13} becomes $\rho_1 \geq \rho_3$, we have
$\rho_1 = \rho_3$~.  Then, taking $\rho_1$ for $\rho_2'$, 
\lref {C123} trivially holds while \lref {C22'} follows from \lref {T123}. 
\qed
\end{proof}

The next lemma connects the monotonicity of the variance-checking
judgment (checking variance at a lower context provides more
information, and is therefore harder) and the anti-monotonicity of the
decomposability formula (decomposing to a higher context provides more
information, and is therefore harder): for a fixed type expression,
the contexts at which you can check variance are higher than the
contexts at which you can decompose. This property, however, only
holds for non-trivial decomposability results (otherwise any context
can decompose): we must decompose from a $v$ to a $v'$ that do not
verify $v \geq v'$, and no variable of the typing context must be
irrelevant.

\begin{lemma}
\label{lem/order-decomposition-variance}
Let $\app{\tau}{\ba}$ be a type and $v$ and $v'$ be variances such that $v
\ngeq v'$.
If ${\G \Der \tau : v \vto v'}$ and $\Delta$ is the principal
context such that $\Delta \der \tau : v$, then, for each non-irrelevant
variable $\alpha$ of $\Delta$, we have $\G(\alpha) \leq \Delta(\alpha)$.
\end{lemma}

\begin{proof}{}
\locallabelreset
Without loss of generality, we consider the case where $\tau$ has
a single, non-irrelevant variable $\alpha$, and $\G$ and $\Delta$ are
singleton contexts over a single variance, respectively $w_\G$ and
$w_\Delta$, with $w_\Delta \neq \rel\Join$.

Therefore, we assume ${(w_\G\alpha) \Der \tau : v \vto v'}$~\llabel
{H1} and 
$(w_\Delta\alpha) \der \tau : v$~\llabel {H2}. 
We  prove that $w_\G \leq w_\Delta$. 
We actually show that for any
$\rho_1, \rho_2$ such that $\rho_1 \prec_\Delta \rho_2$ we also have
$\rho_1 \prec_\G \rho_2$~\llabel C.

From $\rho_1 \prec_\G \rho_2$, we can deduce
$\app{\tau}{\rho_1} \prec_v \app{\tau}{\rho_2}$~\llabel {T12}. 
Applying \lref {H1}, 
we get a $\rho'$ such that $\rho_1 \prec_\G \rho'$~\llabel r and
$\app{\tau}{\rho'} \prec_{v'} \app{\tau}{\rho_2}$~\llabel {Tp2}. We then reason by case
analysis on $v \ngeq v'$, considering the different cases
$\{(\virr, \any), (\vplus, \vminus), (\vminus, \vplus), (\any, \veq)\}$.

If $v$ is $\rel\virr$, the most general  $w_\Delta$ is
$\virr$, a case we explicitly ruled out: there is
nothing to prove.

If $v'$ is $\rel\veq$, then \lref {Tp2} implies 
$\app{\tau}{\rho'} = \app{\tau}{\rho_2}$, and in particular
$\rho' = \rho_2$; our goal \lref C follows from \lref r.

If $(v,v')$ is $(\vplus,\vminus)$ (we won't repeat the
symmetric case $(\vminus,\vplus)$), then \lref {T12} and \lref {Tp2}
becomes
$\app{\tau}{\rho_1} \leq \app{\tau}{\rho_2}$ and 
$\app{\tau}{\rho_2}  \leq\app{\tau}{\rho'}$.
Given \lref r,  the intermediate value lemma (Lemma
\ref{intermediate_value_lemma}) 
ensures the existence of a $\rho''$ such that $\rho \prec_\G \rho''
\prec_\G \rho'$ and $\app{\tau}{\rho''} = \app{\tau}{\rho_2}$. From there, we
deduce $\rho'' = \rho_2$, our goal \lref C follows from \lref r, as in the
previous case.
\qed
\end{proof}

Finally, the following auxiliary lemmas will be useful in the proof of
completeness.

\begin{lemma}
  \label{lem/non-irrelevant-equality}
  If the principal variance $w$ such that 
  $(w\alpha) \der \app{\tau}{\alpha} : v$ holds is not the irrelevant variance
  $\virr$, then
  $\app\tau{\rho_1} = \app\tau{\rho_2}$ implies $\rho_1 = \rho_2$.
\end{lemma}
\begin{proof}{}
  Whatever $v$ is, $\app{\tau}{\rho_1} = \app{\tau}{\rho_2}$ implies
  both $\app{\tau}{\rho_1} \prec_v \app{\tau}{\rho_2}$ and its converse
  $\app{\tau}{\rho_2} \prec_v \app{\tau}{\rho_1}$. 
This holds in particular for the principal variance  $w$ such that 
$(w\alpha) \der \app{\tau}{\alpha} : v$. 
Moreover, by principal inversion (Corrolary~\ref{cor/principal-inversion})
applied twice, we have both $\rho_1 \prec_w \rho_2$ and $\rho_2 \prec_w
\rho_1$. 
If $w$ is distinct from $\virr$ this implies $\rho_1 = \rho_2$.
\qed
\end{proof}

\begin{lemma}
\label{lem/principal-equal-decomposability}
If $\Gamma \Der \tau : v \vto \rel\veq$ holds for some $\Gamma$, and
$\Delta$ is the minimal context such that $\Delta \der \tau : v$
holds, then $\Delta \Der \tau : v \vto \rel\veq$ also hold.
\end{lemma}
\begin{proof}{}
\locallabelreset
Assume $\Gamma \Der \tau : v \vto \rel\veq$, \ie. 
$$
    \forall \br \bs', \app{\tau}{\br} \prec_v \bs' \implies
      \exists \br',\uad \br \prec_\G \br' \wedge \app{\tau}{\br'} = \sigma'
\eqno \llabel{Hdecomp}
$$
Assume that $\Delta$ is minimal for  $\Delta \der \tau : v$~\llabel D.
We show  $\Delta \Der \tau : v \vto \rel\veq$, \ie. 
$$
    \forall \br \bs', \app{\tau}{\br} \prec_v \bs' \implies
      \exists \br',\uad \br \prec_\Delta \br' \wedge \app{\tau}{\br'} =
      \sigma'
\eqno \llabel 0
$$
Let $\br$, $\bs'$ be such that have 
$\app{\tau}{\br} \prec_v \sigma'$~\llabel{Hprec}.  
By~\lref{Hdecomp}, there exists
$\br'$ such that $\app{\tau}{\br'} = \sigma'$~\llabel =.
To prove~\lref 0, it only remains to prove that  $\br \prec_\Delta
\br'$~\llabel C.  
Given~\lref =, the inequality~\lref {Hprec} becomes $\app{\tau}{\br}
\prec_v \app{\tau}{\br'}$~\llabel 4.  Then \lref C follows 
by  principal inversion (Corrolary~\ref{cor/principal-inversion})
applied to \lref 4, given \lref D. 
\qed
\end{proof}

\begin{lemma}
\label{lem/principal-context-for-eq-or-irr}
Let $\Delta$ be the minimal context such that $\Delta \der \tau : v$
holds. If $v$ is $\rel=$, then only variances $\rel\veq$ or $\rel\virr$ may
appear in $\Delta$.  If $v$ is $\virr$, then only $\rel\virr$ may appear in 
$\Delta$.
\end{lemma}
\begin{proof}{}
\locallabelreset
  If $v$ is $\rel\virr$, the context $\Gamma$ with all variances set
  to $\virr$ satifies $\Gamma \der \tau : v$ (as $\sem{\G \der
    \tau : v}$ holds). By principality we have $\Delta \leq
  \Gamma$, so $\Delta$ also has only irrelevant variances as $\rel\virr$
  is the minimal variance.

  If $v$ is $\rel\veq$, we handle each variable of the context
  independently, that is we can assume, \wlg.\XXX{Need justification}, that
  $\tau$ has only one variable $\alpha$. So $\Delta$ is of the form
  $(w\alpha)$ and we know that for any $\rho$, $\rho'$ such that $\rho
  \prec_w \rho'$ we have $\app{\tau}{\rho} = \app{\tau}{\rho'}$~\llabel H. 
If $w$ is not $\rel\virr$, by lemma~\ref{lem/non-irrelevant-equality}
applied to \lref H, we have $\rho = \rho'$ for any $\rho \prec_w \rho'$,
which means that $w$ is $\rel\veq$.
Summing up, we have shown that $w$ is either $\rel\virr$ or $\rel\veq$.
Reasoning similarly in the general case, any variance $w$ of $\Delta$ is
either $\rel\virr$ or $\rel\veq$.
\qed
\end{proof}

We can now prove the converse of the zip soundness (Lemma
\ref{lem/zip-soundness}) that is the core of the future proof of
completeness of the decomposability judgment $\G \der \tau : v \To \vprime{}$.
\begin{theorem}[Zip completeness] 
\label{thm/zip-completeness}
Given any context $\G$, a family of type expressions
$(\app{T_i}{\bs})_\iI$ and a family of variances $(v_i)_\iI$, if the simultaneous
decomposition $\G \Der (T_i : v_i \vto \vprime{i})_\iI$ holds, then there
exists a family of contexts $(\G_i)_\iI$ such that 
$\G \leq \mathop{\zip}_\iI \G_i$ and 
both $\G_i \der T_i : v_i$ and $\G_i \Der T_i : v_i \vto \vprime{i}$
hold for all $i$.

If furthermore $\G \der T_i : v_i$ holds for all $i$, then
$\mathop{\zip}_\iI \G_i$ is precisely $\G$.
\end{theorem}

\begin{proof}{}
\locallabelreset
Let us assume the simultaneous decomposition ${\G \Der (T_i : v_i
\vto \vprime{i})_\iI}$, which expands to: 
$$
\forall \bs', \br,\;
   (\forall i, \app{T_i}\br \prec_{v_i} \sigma'_i)
   \implies 
   \exists \br',\uad
    \br \prec_{\G} \br' \wedge
    (\forall i, \app{T_i}{\br'} \precprime{i} \sigma'_i)
\eqno \llabel{Hyp}
$$
We construct a family of contexts $(\G_i)_\iI$ such that the following holds:
\begin{mathpar}
\G \leq \mathop{\zip}_i \G_\iI 
~\llabel {zip}

\forall i,\,\G_i \der T_i : v_i
~\llabel {ResZip}

\forall i, \G_i \Der T_i : v_i \vto \vprime{i}
~\llabel {ResDecomp}
\end{mathpar}
where \lref {ResDecomp} is equivalent to
$$
\forall i,\uad  \forall \sigma'_i, \br,\uad
   \bigparens {
   \app{T_i}\br \prec_{v_i} \sigma'_i
   \implies
   \exists \br',\uad 
        \br \prec_{\G_i} \br'
        \wedge
        \app{T_i}{\br'} \precprime{i} \sigma'_i
   }
\eqno\llabel{ResDecomp'}
$$
The first step is to move from the entailment~\lref {Hyp} of the form
$\forall \br, (\forall \iI, \dots) \implies (\exists \br', (\forall \iI, \dots))$
to the weaker form 
$(\forall \iI, \forall \br, \dots \implies \exists \br', \dots)$, 
but closer to \lref {ResDecomp'}.
More precisely, we show that 
\[
  \forall i, \uad \forall \sigma'_i,\br,\uad
  \bigparens {
  \app{T_i}\br \prec_{v_i} \sigma'_i \implies
  \exists\br',\uad
  \br \prec_{\G} \br'  \wedge
   \app{T_i}{\br'} \precprime{i} \sigma'_i 
  }
  \eqno\llabel{EasyPart}
\]
that is, $\forall i, {\G \Der T_i : v_i \vto \vprime{i}}$~\llabel {Easy}.
Let $i$, $\sigma'_i$, and $\br$ be such that 
$\app{T_i}{\br} \prec_{v_i} \sigma'_i$~\llabel {Tisi'}.  
We show that there exists
a $\br'$ such that $\app{T_i}{\br} = \sigma'_i$~\llabel {Ti=si'} and
$\br \prec_{\G} \br'$~\llabel{r<Gr'}. 
Let us extend our type $\sigma'_i$ to
a family $\bs'$,  defined by taking
 $\sigma'_j$ equal to $\app{T_j}{\br}$ for $j$ in $I \setminus \set i$. 
By construction, we have
$(\forall \jI, \app{T_j}{\br} \prec_{v_j} \sigma'_j)$.
Therefore, we may apply~\lref{Hyp} to get a $\br'$ such that
$\br \prec_{\G} \br'$, \ie. our first goal \lref {r<Gr'}, and 
$(\forall j, \app{T_j}{\br'} = \sigma'_i)$, which implies our second goal
\lref {Ti=si'} when $j$ is $i$. This proves \lref{EasyPart}. 

We now prove that we can refine this to have $\br \prec_{\G_i} \br'$
for $(\G_i)_\iI$ such that $\G \leq \mathop{\zip}_i \G_i$. 

Let $\D_1$ and $\D_2$ be the minimal contexts such both that $\D_1 \der T_1
: v_1$ and $\D_2 \der T_2 : v_2$ hold~\llabel D.  
Let \lref {zip}', \lref {ResZip}', and \lref {ResDecomp}'  be
obtained by replacing $\Gamma_i$'s by $\Delta_i$'s in our three goals \lref
{zip}, \lref {ResZip}, and \lref {ResDecomp}.  In fact \lref {ResZip}' is just
\lref D.
By Lemma~\ref{lem/principal-equal-decomposability} applied
to~\lref {Easy} twice, given~\lref D, we have both
$\D_1 \Der T_1 : v_1 \vto \vprime{1}$ and
$\D_2 \Der T_2 : v_2 \vto \vprime{2}$, that is, \lref {ResDecomp}.

Hence $\D_1$ and $\D_2$ are correct choices for $\G_1$ and $\G_2$
if they also satisfy the goal~\lref{zip}', \ie. $\D_1 \zip \D_2
\geq \G$. 
We now study when remaining goal \lref {zip}' holds and, when it does not,
propose a different choice for $\G_1$ and $\G_2$ that respect all three
goals.

\Wlg., we assume that $I$ is reduced to $\set {1, 2}$ and that 
there is only one free variable $\beta$ in $T_1, T_2$.
Since we focus on a single variable of the context we name $w_1$ and $w_2$
the variances of $\beta$ in $\D_1$ and $\D_2$, \XXX{Changed $\G_i$ to
$\D_i$} respectively. We now reason by case analysis on the variances $w_1$
and $w_2$.

If both of them are $\virr$, we have $\D_1 \zip \D_2 = (\virr\beta)$,
so we do not necessarily have $\G \leq \D_1 \zip \D_2$. 
Instead, we make a different choice for $G_2$. Namely, we  pick $\G$ for
$\G_2$ and keep $\D_1$ for  $\G_1$.  
As $\D_2$ is $\virr$ we have $\D_2 \leq \G_2$, so by
monotonicity of the variance checking judgment we have $\G_2 \der
T_2 : v_2$ from~\lref D, and we still have $\G_2 \Der T_2 : v_2 \vto
\rel\veq$ from~\lref{Easy}. Hence~\lref {ResZip} and \lref {ResDecomp}
are reestablished.
Finally, we have $\G_1 \zip \G_2 = \G$, so in
particular $\G \leq \G_1 \zip \G_2$, \ie. \lref {zip}. 

If only one of the $w_i$ is $\virr$, we may assume, \wlg., that it is
$w_1$.  Then $\D_1 \zip \D_2$ is $\D_2$ and we only need to show that 
$\G \leq \D_2$~\llabel {GD2}. 
From \lref {Easy}, we have $\G \Der T_2 : v_2 \To
\vprime{2}$. We then make a case analysis on $v_2$:
if $v_2$ is not $\rel\veq$, then 
by 
since $\D_2$ is most general and $w_2 \neq \virr$, we may apply
Lemma~\ref{lem/order-decomposition-variance} 
to get \lref {GD2};
Otherwise, $v_2$ is $\rel\veq$;
Lemma~\ref{lem/principal-context-for-eq-or-irr} applied to \lref D implies
that $\D_2$ is itself $(\veq\beta)$, and \lref {GD2} trivially
holds. \XXX{Because $\rel\veq$ is the top variance. Recheck as this was not
clearly the argument used.}

Finally, if none of the $w_i$ is $\virr$, we first prove that they are
both $\rel\veq$. In fact, we only prove that $w_1$ is $\rel\veq$~\llabel
{w1}, as the other case follows  by symmetry. 
To prove~\lref {w1}, we assume that $\rho''$ be such that $\rho
\prec_{w_1} \rho''$  and we show that $\rho = \rho''$~\llabel {LC} holds.  
By~\lref D, we have $\app{T_1}{\rho} \prec_{v_1} \app{T_1}{\rho''}$. 
By reflexivity, we have $\app{T_2}{\rho} \prec_{v_2} \app{T_2}{\rho}$.
We can use those
two inequalities to invoke our simultaneous decomposability
hypothesis~$\lref{Hyp}$ with $\app{T_1}{\rho''}$ for $\sigma'_1$ 
and $\app{T_2}{\rho}$ for $\sigma'_2$ 
to get a $\rho'$ such that both
$\app{T_1}{\rho'} \precprime{1} \app{T_1}{\rho''}$ and 
$\app{T_2}{\rho'} \precprime{2} \app{T_2}{\rho}$ hold.
By Lemma~\ref {lem/non-irrelevant-equality}
applied with \lref D, this implies both $\rho' = \rho''$ and $\rho' = \rho$, 
and therefore~\lref {LC}. 

Therefore, the only remaining case is when $w_1$ and $ w_2$ are both
$\rel\veq$. Then $\D_1 \zip \D_2$ is $(\veq\beta)$, which is the highest
single-variable context. So our goal \lref {zip}' trivially holds. 

Of these several cases, one ($\virr,\virr$) has $\G_1 \zip \G_2 =
\G$ directly, and in the others $\G_1$ and $\G_2$ were defined as the
minimal contexts such that $\G_1 \der T_1 : v_1$ and $\G_2 \der
T_2 : v_2$. If we add the further hypothesis that for each $i$, $\G
\der T_i : v_i$ holds, then by principality of the $\G_i$, we have that
$\G_i \leq \G$ for each $i$. This implies that we have
$(\mathop{\zip}_\iI \G_i) \leq \G$ (when it is defined, $\zip$
coincides with the lowest upper bound $\wedge$).
By combination with \lref {zip}, we  get  $(\mathop{\zip}_\iI \G_i) = \G$.
\qed
\end{proof}

\begin{lemma}[Completeness of syntactic decomposability]
\label{lem/decomposability-completeness}\noindent\\
If $\sem{\G \der \tau : v \To v'}$ holds for $v' \in \{\veq,\virr\}$,
then $\G \der \tau : v \To v'$ is provable.
\end{lemma}
\begin{proof}{}
\locallabelreset
Assume $\sem{\G \der \tau : v \To v'}$ holds for $v' \in \{\veq,\virr\}$,
\ie. 
$\G \der \tau : v$~\llabel d
and
$\G \Der \tau : v \vto v'$~\llabel D, which 
expands to
$$
\let \sigma \tau
\forall (\br : \G), \sigma',\uad
      \app \sigma \br \prec_{v} \sigma' \implies
      \exists (\br' : \G),\uad
        \br \prec_\G \br'
        \wide\wedge \app \sigma {\br'} \prec_{v'} \sigma'
\eqno \llabel {D'}
$$
We show  $\G \der \tau : v \To v'$~\llabel C  by 
structural induction on $\tau$

If $v \geq v'$ holds, then~\lref C
directly follows from Rule \Rule{sc-Triv}.
This applies in particular when $v' = \virr$. 
Hence, we only need to consider the remaining cases where
$v'$ is $\rel\veq$ and $v \not \geq v'$


We now reason by cases on $\tau$. 

\Case {Case $\tau$ is a variable $\alpha$.} \lref {D'} becomes
$$
\forall \rho, \tau',\uad
      \rho \prec_{v} \tau' \implies
      \exists \rho',\uad
        \rho \prec_\G \rho'
        \wide\wedge \rho' = \tau'
$$
This means that if $\rho \prec_v \tau'$ holds then $\rho \prec_{\G} \tau'$ also
holds: the variance $w\alpha \in \G$ satisfies $v \geq w$. 
Since, the hypothesis \lref d implies $v \leq w$, we have
$v = w$. Therefore, \lref C follows by Rule \Rule{sc-Var}.

\Case {Case $\tau$ is of the form $\ty{\bs}{t}$.}
By inversion, the derivation of~\lref d must end with rule \Rule
{vc-Constr}, hence we have  $\G \der \sigma_i : \varcomp v
{w_i}$~\llabel{HypVar} for each $i \in I$ with  $\G \der
\mc{type}~\ty{\bar{w\alpha}}{t}$~\llabel t.   
 

  \newcommand{\Iirr}{{I_{\virr}}}
  \newcommand{\Inirr}{I_{\not \virr}}

Let us show that $\G \Der (\sigma_i : \varcomp v {w_i} \vto {\varcomp \veq
{w_i}})_\iI$~\llabel S, \ie. 
$$
    \forall \br, \bs',\uad
      \parens 
        {\forall \iI, \app{\sigma_i}{\br} \prec_{v_i} \sigma'_i}
      \implies
      \exists \br',\uad \br \prec_\G \br'
        \wedge \parens 
         {\forall \iI, \app {\sigma_i} {\br'} \prec_{v'_i} \sigma'_i}
$$
Let $\br$ and $\bs'$ be such that  
$\app{\sigma_i}{\br} \prec_{v_i} \sigma'_i$ holds for all $i$ in $I$. 
From this and \lref {HypVar}, we have $\app {(\ty\bs t)} \br \prec_{v}
\ty{\bs'} t$. By applitcation of \lref D, there exists
$\br'$ such that $\br \prec_\G \br'$ and
$\app {(\ty\bs t)} {\br'} \prec_{\veq} \ty{\bs'} t$. 
By inversion of subtyping\XXX{Attention, not by the inversion lemma...}
\XXX{This is just mentioned in plain text, but it be a lemma},
this implies $\app{\sigma_i}{\br'} \prec_{\varcomp \veq {w_i}} \sigma'_i$, 
for all $i$ in $I$. This proves~\lref S.
We also note that the constructor $\tyc{t}$ is $v$-closed~\llabel {Note}.

To prove our goal~\lref C, we construct a family $(\G_i)_\iI$ of contexts that satisfies
$\mathop{\zip}_\iI \G_i = \G$~\llabel G and subderivations $\G_i \der
  \sigma_i : \varcomp v {w_i} \To \varcomp \veq {w_i}$~\llabel i, since then
the conclusion~\lref C follows by an application of
rule \Rule {sc-Constr} with~\lref t, \lref G, and \lref i.

We will handle separately the arguments $\sigma_j$ that are irrelevant, \ie.
when $w_j$ is $\virr$,  from the rest. 
Let $\Iirr$ be the set of indices with irrelevant variances and $\Inirr$ the
others. 
 
For any $i \in \Iirr$, $\varcomp v {w_i}$ and $\varcomp {\veq} {w_i}$ are
both $\virr$ so the condition~\lref i, which becomes $\G_i \der \sigma_i :
\virr \To \virr$, is void of content ($\app{\sigma_i}{\br}
\prec_{\virr} \sigma'_i \implies \exists \br',\uad \app{\sigma_i}{\br'}
\prec_{\virr} \sigma'_i$ is always true). More precisely, let $\Gamma_i$ be
the irrelevant context having only irrelevant variances.  Then \lref i
follows by Rule \Rule {sc-Triv}.

Since the decomposability constraints for $i \in \Iirr$ such that
$w_i = \virr$ are trivial, \lref S  is equivalent to $\G \Der (\sigma_i : \varcomp
v {w_i} \vto {\varcomp \veq {w_i}})_{i \in \Inirr}$~\llabel {Snirr}. 
\XXX{Better show it, formally.}

For each $i \in \Inirr$, $\rel{\varcomp \veq {w_i}}$ equals $\rel\veq$, so
\lref {Snirr} becomes $\G \Der (\sigma_i : \varcomp v {w_i} \vto
\rel\veq)_{i \in \Inirr}$.
By zip completeness (Theorem~\ref{thm/zip-completeness}), there is a family
$(\Gamma_i)_{i \in \Inirr}$ such that
$\G \leq \mathop{\zip}_{i \in \Inirr} \Gamma_i$ and
both $\G_i \der \sigma_i : \varcomp v {w_i}$ and $\G_i \Der \sigma_i :
\varcomp v {w_i} \vto \rel\veq$, \ie. 
$\sem{\Gamma_i \der \sigma_i :\varcomp v {w_i} \To \rel\veq}$~\llabel h hold for any
$i \in \Inirr$, 
Furthermore, since we also
have~$\lref{HypVar}$, we can strengthen our result into $\mathop{\zip}_{i \in \Inirr} 
\Gamma_i = \G$.  By induction hypothesis applied to \lref h, we have 
\lref i for $i \in \Inirr$. 

We have two families of contexts over domains $\Iirr$ and $\Inirr$ that
partition $I$; we can union them in a family $(\Gamma_i)_{\iI}$ that has
subderivations $\forall \iI, \G_i \der \sigma_i : \varcomp v {w_i} \To
\varcomp {v'} {w_i}$. As the contexts in $\Iirr$ are all irrelevant, they
are neutral for the zipping operation: $\mathop{\zip}_{i \in I} \Gamma_i$
is equal to $\mathop{\zip}_{i \in \Inirr} \Gamma_i$, that is $\G$. 
This proves \lref G while \lref i has already been proved separately for 
$i \in \Iirr$ and $i \in \Inirr$.  
\qed
\end{proof}

\begin{remark}[Note \lref {Note}]
  The reason why the head constructor $\tyc{t}$ is closed in our
  system is a bit subtle. The statement of $v$-closure of
  $\ty{\ba}{t}$ can be formulated in term of decomposability $\Gamma
  \Der \ty{\ba}{t} : v \vto \rel\veq$. It is very close from our
  decomposability hypothesis $\Gamma \Der \ty{\bs}{t} : v \vto
  \rel\veq$, but uses variables $\ba$ instead of full type expressions
  $\bs$. In the general case, it is in fact a stronger result: we
  cannot logically derive it from our hypothesis. However, the shape
  of our subtyping relation (even with the non-atomic rule
  \Rule{sub-PQ}) verifies the property that subtyping is only
  determined by the head constructors (with variance conditions on the
  parameters as a whole), which allows to derive this property.

  All the subtyping systems we are aware of in the literature also
  satisfy this property, which to our knowledge has no name, but it is
  still a potential limit on the applicability of our completeness
  result. One could imagine a binary type $\ty{(\alpha,\beta)}{weird}$
  such that $\ty{(\alpha,\tyc{int})}{weird}$, as a type expression
  depending only on $\alpha$, is upward-closed, while
  $\ty{(\alpha,\beta)}{weird}$ in general is not, because of
  a specific subtyping relation $\ty{(\alpha,\tyc{bool})}{weird} \leq
  \tyc{q}$ for example. Such non-uniform subtyping relations could
  require a different analysis, but to our knowledge they are not
  present in any of the programming language for which we might want
  to apply this work.
\end{remark}

\subsection{Back to the correctness criterion}

\locallabelreset

Remember the correctness criterion: 
\[
\forall \bs, \bs', \br, \quad
\left(
   \ty{\bs}{t} \leq \ty{\bs'}{t}
     \wedge \app D{\bs,\br}
   \implies
   \exists \br',\; \app D{\bs',\br'}
     \wedge \app\tau\br \leq \app\tau{\br'}
\right) 
\eqno (\Rule {Req})
\]
We now show how the closure judgment $\G \der \tau : v \To v'$ can be
used to verify that this criterion holds: we will express this
criterion in an equivalent form that uses the interpretation of our
judgments.

The first step is to rewrite the property $\ty{\bs}{t} \leq
\ty{\bs'}{t}$ using the variance annotation $\Gva$ of
$\tyc{t}$. Again, we are taking the variance annotation for the
datatype $\tyc{t}$ as granted (this is why we can use it in this
reasoning step), and checking that the definitions of the constructors
of $\tyc{t}$ are sound with respect to this annotation. \[
\forall \bs, \bs', \br, \quad
\left(
   (\forall i, \sigma_i \prec_{v_i} \sigma'_i)
      \wedge \app D{\bs,\br}
   \implies
   \exists \br',\; \app D{\bs',\br'}
     \wedge \app\tau\br \leq \app\tau{\br'}
\right) 
\eqno \llabel 2
\]
Since, the constraint $\app D{\ba,\bb}$ is a set of equalities of the form
$\alpha_i = \app{T_i}\bb$ (where $T_i$ is a type), \lref 2 is actually:
\[
\forall \bs, \bs', \br, \uad
   \parens {\forall i, \sigma_i \prec_{v_i} \sigma'_i}
     \wedge \parens {\forall i, \bs_i = \app{T_i}{\br}}
   \implies
   \exists \br',\; (\forall i, \app{T_i}{\br'} = \sigma'_i)
     \wedge \app\tau\br \leq \app\tau{\br'}
\]
Substituting the equalities and, in particular, removing the quantification
on the $\bs$, which are fully determined by the equality constraints $\bs =
\app{\bar{T}}\br$, we get:
\[
\forall \bs', \br, \quad
   (\forall i, \app{T_i}\br \prec_{v_i} \sigma'_i)
   \implies
   \exists \br',\; (\forall i, \app{T_i}{\br'} = \sigma'_i)
     \wedge \app\tau\br \leq \app\tau{\br'}
\eqno\llabel 3
\]
By inversion (Theorem~\ref{thm/inversion}), we may replace the goal $\app \tau
\br \leq \app \tau {\br'}$ by the formula
$\exists \G, (\G \der \tau : \vplus) \wedge (\br \prec_\G \br')$.
Moreover, since $\G \der \tau : \vplus$ always for for some $\G$, we may 
move  this quantification in front. Hence, \lref 3  is equivalent to:
\[
\exists \G, \uad\bigwedge 
\begin{cases}
\G \der \tau : \vplus
\\
\forall \bs', \br,\;
   (\forall i, \app{T_i}\br \prec_{v_i} \sigma'_i)
   \implies
   \exists \br',\; (\forall i, \app{T_i}{\br'} = \sigma'_i)
     \wedge \br \prec_{\G} \br'
\end{cases}
\eqno\llabel 4
\]
We may recognize in second clause the simultaneous 
decomposability judgment (Definition~\ref{def/decomposability})
$\G \der (T_i : v_i \vto \veq)_{\iI}$. Hence, \lref 4 is in fact: 
$$
\exists \G, \uad
\G \der \tau : \vplus
\wide\wedge \G\der (T_i : v_i \vto \veq)_{\iI}
\eqno\llabel {4b}
$$
Then, comes the delicate step of this series of equivalent rewriting:
\[
\exists \G, (\G_i)_\iI, \uad
\bigwedge
\begin{cases}
\G \der \tau : \vplus
\\
     \G = \mathop{\zip}_\iI \G_i
     \wide\wedge
      \forall i, \uad
      \bigparens {\G_i \der T_i : v_i
      \wide\wedge
      \G_i \der T_i : v_i \vto \veq}
\end{cases}
\eqno\llabel 5
\]
The reverse imiplication from \lref 5 to \lref {4b} is is the zip soundness
(Lemma \ref{lem/zip-soundness}).  

The direct implication, from \lref {4b} to
\lref 5 is more involved: let $\G_0$ be such that $\G_0 \der \tau
: \vplus$.  By zip completeness (Theorem
\ref{thm/zip-completeness}), with the hypotheses of $\lref{4}$, there
exists a family $(\G_i)_\iI$ satisfying the typing, zipping and
decomposability of second line of \lref{5} with $\G_0 \leq \mathop{\zip}_\iI
\G_i$. We take $\mathop{\zip}_\iI \G_i$ for $\G$. Then, from $\G_0 \leq \G$
we get $\G \der
\tau : \vplus$ by monotonicity (Lemma~\ref {lem/monotonicity}).

As a last step, the last conjuncts of \lref 5 are equivalent to $\forall
i,\,\G_i \der T_i : v_i \To \rel\veq$ by interpretation of syntactic
decomposability (Definition~\ref {def/cc/semantics}) and soundness and
completeness of zipping (lemmas~\ref {lem/decomposability-soundness}
and~\ref {lem/decomposability-completeness}).  Therefore, \lref 5 is
equivalent to:
\[
\exists \G, (\G_i)_\iI,\quad
  \G \der \tau : \rel\vplus
  \ \ \wedge\ \ 
  \G = \mathop{\zip}_\iI \G_i
  \ \ \wedge\ \ 
  \forall \iI,\,\G_i \der T_i : v_i \To \rel\veq
  \qquad
\]
which is our final criterion. 

\end{version}

\begin{theorem}[Algorithmic criterion]
  Given a variance annotation $(v_i\alpha_i)_\iI$ and a constructor
  declaration of type
  $(\exists \bb\left[\mathop{\bigwedge}_\iI \alpha_i = \app{T_i}{\bb}\right].\,\app \tau \bb)$,
  the soundness criterion \Rule{Req} for this constructor is
  equivalent to
$$
\exists \G, (\G_i)_\iI,\quad
  \G \der \tau : \rel\vplus
  \ \ \wedge\ \ 
  \G = \mathop{\zip}_\iI \G_i
  \ \ \wedge\ \ 
  \forall \iI,\,\G_i \der T_i : v_i \To \rel\veq
$$
\end{theorem}

The three parts of this formula can be explained to a user, as soon
as the underlying semantic phenomenons (variable interference through
zipping, and upward- and downward-closure) have been
understood---there is no way to get around that. They are best read
from right to left. The last part on the $(T_i)_\iI$ is the
decomposability requirement that failed in our example with \code{<
  m : int >}: the type expressions equated with a covariant variable
should be upward-closed, and those equated with a contravariant one downward-closed. The
zipping part checks that the equations do not create interference
through shared existential variables, as in \lstinline{type ($\vplus\alpha$, $\veq\beta$) eq = Refl of $\exists \gamma [\alpha = \gamma, \beta = \gamma]$}. Finally, the variance check corresponds to the classic
variance check on argument types of \ADTs. One can verify that in
presence of a simple \ADT, this new criterion reduces to the simple
syntactic criterion.

\Esop{}{\paragraph{Pragmatic evaluation of this criterion}}

This presentation of the correctness criterion only relies on
syntactic judgments. It is pragmatic in the sense that it suggests
a simple and direct implementation, as a generalization of the check
currently implemented in type system engines---which corresponds to the
$\Gamma \der \tau : \rel\vplus$ part.

To compute the contexts $\Gamma$ and $(\Gamma_i)_\iI$ existentially
quantified in this formula, one can use a variant of our syntactic
judgments where the environment $\Gamma$ is not an input, but an
output of the judgment; in fact, one should return for each variable
$\alpha$ the \emph{set} of possible variances for this judgment to
hold. For example, the query $(? \der \alpha * \ty{\beta}{ref}: \vplus)$
should return $(\alpha \mapsto \{\vplus, \veq\}; \beta \mapsto
\{\veq\})$. Defining those algorithmic variants of the judgments is
routine. 
The sets of variances corresponding to the decomposability of the
$(T_i)_\iI$ ($? \der T_i : v_i \To \rel\veq$) should be zipped
together and intersected with the possible variances for $\tau$,
returned by
($? \der \tau : \vplus$). The algorithmic criterion is satisfied if
and only if the intersection is not empty; this can be decided in
a simple and efficient way.

\section{Discussion}
\label{sec/discussion}

\subsection{Upward and downward closure in a ML type system}

In the type system we have used so far, all type constructors but $p$
and $q$ are both upward and downward-closed.  This simple situation,
however, does not hold in general: richer subtyping relations will
have weaker invertibility properties. As soon as a bottom type $\bot$
is introduced, for example, such that that for all type $\sigma$ we
have $\bot \leq \sigma$, downward-closure fails for all types -- but
$\bot$ itself. For example, products are no longer downward-closed:
$\G \der \sigma * \tau \geq \bot$ does not implies that $\bot$ is
equal to some $\sigma' * \tau'$. Conversely, if one adds a top type
$\top$, bigger than all other types, then most type are not
upward-closed anymore.

In OCaml, there is no $\bot$ or $\top$ type\footnote{A bottom type
  would be admissible, but a top type would be unsound in OCaml, as
  different types may have different runtime
  representations. Existential types, that may mix values of different
  types, are constructed explicitly through a boxing step.}. 
However, object
types and polymorphic variants have subtyping,  so they are, in general,
neither upward nor downward-closed. Finally, subtyping is also used
in private type definitions, which were demonstrated in the example.
Our closure-checking relation therefore degenerates into the
following, quite unsatisfying, picture:
\begin{itemize}
\item no type is downward-closed because of the existence of private types;
\item no object type  but the empty object type is upward-closed;
\item no arrow type is upward-closed because its left-hand-side would
  need to be downward-closed;
\item datatypes are upward-closed if their components types are.
\end{itemize}
From a pragmatic point of view, the situation is not so bad; as our
main practical motivation for finer variance checks is the relaxed
value restriction, we care about upward-closure (covariance) more than
downward-closure (contravariance). This criterion tells us that
covariant parameters can be instantiated with covariant datatypes
defined from sum and product types (but no arrow), which would satisfy
a reasonable set of use cases.

\subsection{A better control on upward and downward-closure}

There is a subtle design question here. Decomposability is
fundamentally a \texttt{negative} statement on the subtyping relation,
guaranteeing that some types have no supertypes of a different
structure. It is therefore not necessarily preserved by addition to
the subtyping relation -- our system, informally, is
{\tt non\--monotone} in the subtyping relation.

This means that if we adopt the correctness criterion above, we must
be careful in the future not to enrich the subtyping relation too
much. Consider \code{private} types for example: one could imagine
a symmetric concept of a type that would be strictly \emph{above}
a given type $\tau$; we will name those types \code{invisible} types
(they can be constructed, but not observed). Invisible types and \GADT
covariance seem to be working against each other: if the designer adds
one, adding the other later will be difficult.

A solution to this tension is to allow the user to \emph{locally}
guarantee negative properties about subtyping (what is \emph{not}
a subtype), at the cost of selectively abandoning the corresponding
flexibility.  Just as object-oriented languages have \code{final}
classes that cannot be extended any more, we would like to be able to
define some types as \code{downward-closed} (respectively
\code{upward-closed}), that cannot later be made \code{private}
(resp. \code{invisible}). Such declarations would be rejected if the
defining type, for example an object type, already has subtypes
(resp. supertypes), and would forbid further declarations of types
below (resp. above) the defined type, effectively guaranteeing
downward (resp. upward) closure.

Finally, upward or downward closure is a semantic aspect of a type
that we must have the freedom to publish through an interface:
abstract types could optionally be declared \code{upward-closed} or
\code{downward-closed}.

\subsection{Subtyping constraints and variance assignment}
\label{sec/gadts-with-subtyping-constraints}

We will now revisit our example of strongly typed
expressions in the introduction.
\begin{version}{\Not\Esop}
It is written, using the guarded existential notation:
\begin{lstlisting}
  type $\ty{\alpha}{exp}$ =
    | Val of $\exists \beta [\alpha = \beta].\, \beta$
    | Int of $[\alpha = \tyc{int}].\,\tyc{int}$
    | Thunk of $\exists \beta \gamma [\alpha = \gamma].\, \ty{\beta}{exp} * (\beta \to \gamma)$
    | Prod of $\exists \beta\gamma[\alpha = \beta*\gamma].\, \ty{\beta}{exp} * \ty{\gamma}{exp}$
\end{lstlisting}
\end{version}
A simple way to get such a type to be covariant would be, instead of
proving delicate, non-monotonic upward-closure properties on the tuple
type involved in the equation $\alpha = \beta * \gamma$, to
\emph{change} this definition so that the resulting type is obviously
covariant:
\begin{lstlisting}
type $\ty{\vplus\alpha}{exp}$ =
  | Val of $\exists \beta [\alpha \geq \beta].\, \beta$
  | Int of $[\alpha \geq \tyc{int}].\,\tyc{int}$
  | Thunk of $\exists \beta \gamma [\alpha \geq \gamma].\, \ty{\beta}{exp} * (\beta \to \gamma)$
  | Prod of $\exists \beta\gamma[\alpha \geq \beta*\gamma].\, \ty{\beta}{exp} * \ty{\gamma}{exp}$
\end{lstlisting}
We have turned each equality constraint $\alpha = \app T \bb$ into
a subtyping constraint $\alpha \geq \app T \bb$. For a type $\alpha'$
such that $\alpha \leq \alpha'$, we get by transitivity that
$\alpha' \geq \app T \bb$. This means that $\ty{\alpha}{exp}$
trivially satisfies the correctness criterion \Rule {Req}. Formally,
instead of checking $\G \der T_i : v_i \To \rel\veq$, we are now
checking $\G \der T_i : v_i \To \rel\vplus$, which is significantly
easier to satisfy\Esop{}{\footnote{ Note that the formal proofs of the
    precedent section were, in some cases, specialized to the equality
    constraint. More precisely, our decomposability criterion is still
    sound when extended to arbitrary subtyping constraints, but its
    completeness is unknown and left to future work.  }}: when $v_i$
is itself $\vplus$ we can directly apply the \Rule{sc-Triv} rule. Note
that this only works in the easy direction: while
$\G \der T_i : \rel\vplus \To \rel\vplus$ is easy to check,
$\G \der T_i : \rel\vplus \To \rel\vminus$ is just as hard as
$\G \der T_i : \rel\vplus \To \rel\veq$. In particular, an equality
($\sigma = \sigma'$) is already equivalent to a pair of inequalities
($\sigma \leq \sigma' \wide\wedge \sigma \geq \sigma'$).

While this different datatype gives us a weaker
subtyping assumption when pattern-matching, we are still able
to write the classic function
$\mc{eval} : \ty{\alpha}{exp} \to {\alpha}$, because the constraints
$\alpha \geq \tau$ are in the right direction to get an $\alpha$ as
a result.
\begin{lstlisting}
let rec eval : $\ty{\alpha}{exp} \to \alpha$ = function
  | Val $\beta$ (v : $\beta$) -> (v :> $\alpha$)
  | Int (n : int) -> (n :> $\alpha$)
  | Thunk $\beta$ $\gamma$ ((v : $\ty{\beta}{exp}$), (f : $\beta \to \gamma$)) ->
    (f (eval v) :> $\alpha$)
  | Prod $\beta$ $\gamma$ ((b : $\ty{\beta}{exp}$), (c : $\ty{\gamma}{exp}$)) ->
    ((eval b, eval c) :> $\alpha$)
\end{lstlisting}

This variation on \GADTs, using subtyping instead of equality
constraints, has been studied by Emir
\etal.~\cite{csharp-generalized-constraints} in the context of the
\csharp programming language---it is also expressible in
Scala. However, using subtyping constraints in \GADTs has important
practical drawbacks in a ML-like language. While typed object-oriented
programming languages tend to use  explicit polymorphism and
implicit subtyping, ML uses implicit polymorphism and explicit
subtyping (when present).  Thus in ML, equality
constraints can be implicitly used while subtyping constraints 
must be explicitly used: unification-based inference favors
bidirectional equality over unidirectional subtyping. This makes \GADT
definitions based on single subtyping constraints less convenient to
use, because of the corresponding syntactic burden, and this is
probably the reason why the notion of \GADTs found in
functional languages use only equality constraints. Subtyping
constraints need also be explicit in the type declaration, forcing the
user out of the convenient ``generalized codomain type'' syntax.

Finally, weakening equality constraints into a subtyping constraint in
one direction is not always possible; sometimes the strictly weaker
expressivity of the type forbids important uses. One must then use an
equality constraint, and use our decomposability-based reasoning to
justify the variance annotation. Consider the following example:
\begin{lstlisting}[xleftmargin=2em]
type $\ty{\vplus\alpha}{tree}$ =
  | Node of $\exists \beta [\alpha = \ty{\beta}{list}].\,\ty{(\ty{\beta}{tree})}{list}$
let append : $\ty{\alpha}{tree} * \ty{\alpha}{tree} \to \ty{\alpha}{tree}$ = function
  | Node $\beta_1$ (l1 : $\ty{\ty{\beta_1}{tree}}{list}$), Node $\beta_2$ (l2 : $\ty{\ty{\beta_2}{tree}}{list}$) ->
       Node (List.append l1 l2)
\end{lstlisting}

We know that the two arguments of \texttt{append} have the same type
$\ty{\alpha}{tree}$. When matching on the \texttt{Node} constructors,
we learn that $\alpha$ is equal to both $\ty{\beta_1}{list}$ and
$\ty{\beta_2}{list}$, from which we can deduce that $\beta_1$ is equal
to $\beta_2$ by non-irrelevance of \texttt{list}. The concatenation of
the lists \texttt{l1} and \texttt{l2} type-checks because this
equality holds. If we used a type system without the decomposability
criterion, we would need to turn the constructor constraint into
$\exists \beta [\alpha \geq \ty{\beta}{list}]$ to preserve covariance
of $\ty{\alpha}{tree}$ . We wouldn't necessarily have $\beta_1$ and
$\beta_2$ equal anymore, so \texttt{(List.append l1 l2)}, hence the
definition of \texttt{append} would not
type-check\Esop{}{\footnote{The motivated reader may want to work out
    the same example on our $\ty{\alpha}{exp}$ datatype: for an
    \texttt{append} function of type
    $\ty{\ty{\alpha}{list}}{exp}*\ty{\ty{\alpha}{list}}{exp} \to \ty{\ty{\alpha}{list}}{exp}$,
    only the constructors \texttt{Val} and \texttt{Thunk} need to be
    considered. We get the same result: an equality constraint is
    necessary to express \texttt{append}.}}. We would need
decomposability-based reasoning to deduce, from
$\alpha \geq \ty{\beta}{list}$ and the fact that $\tyc{list}$ is
upward-closed, that in fact $\alpha = \ty{\beta'}{list}$ for some
$\beta'$.

This demonstrates that single subtyping constraints and our novel
decomposability check on equality constraints are of incomparable
expressivity: each setting handles programs that the other cannot
type-check. From a theoretical standpoint, we think there is value in
exploring the combination of both systems: using subtyping constraints
rather than equalities, but also using decomposability to deduce
stronger equalities when possible.

Note that while our soundness result directly transposes to
a type-system with decomposability conditions on subtyping rather than
equality constraints, our completeness result is special-cased on
equality constraints. Completeness in the case of subtyping
constraints is an open question.

\section*{Related Work}


Simonet and Pottier \cite{simonet-pottier-hmg-toplas} have studied
\GADTs in a general framework HMG(X), inspired by HM(X). They were
interested in type inference using constraints, so considered
\GADTs with arbitrary constraints rather than type equalities, and
considered the case of subtyping with applications to information
flow security in mind. 
Their formulation of the checking problem for datatype
declarations, as a constraint-solving problem, is
exactly our semantic criterion and is not amenable to a direct
implementation. Correspondingly, they did not encounter any of the
new notions of upward and downward-closure and variable interference
(zipping) discussed in the present work.
They define a dynamic semantics and prove that this semantic criterion
implies subject reduction and progress. However, we cannot directly
reuse their soundness result as they work in a setting where all
constructors are upward- and downward-closed (their subtyping relation
is atomic). We believe this is only an artifact of their presentation
and their proof should be easily extensible to our setting.


\medskip 

Emir, Kennedy, Russo and Yu~\cite{csharp-generalized-constraints}
studied the soundness of an object-oriented calculus with subtyping
constraints on classes and methods. Previous work \cite{gadt-and-oop}
had established the correspondence between equality constraints on
methods in an object-oriented style and GADT constraints on type
constructors in functional style. Through this surprisingly
non-obvious correspondence, their system matches our presentation of
\GADTs with subtyping constraints and easier variance assignment,
detailed in~\S\ref{sec/gadts-with-subtyping-constraints}. They provide
several usage examples and a full soundness proof using a classic
syntactic argument. However, they do not consider the more delicate
notions of decomposability, and their system therefore cannot handle
some of the examples presented here.

\subsection*{Future Work}

\begin{version}{\Not\Esop}
\paragraph{On the relaxed value restriction}

Regarding the relaxed value restriction, which is our initial
practical motivation to investigate variance in presence of \GADTs,
there is also future work to be done to verify that it is indeed
compatible with this refined notion of variance. While the syntactic
proof of soundness of the relaxation doesn't involve subtyping
directly, the ``informal justification'' for value restriction uses
the admissibility of a global bottom type $\bot$ to generalize
a covariant unification variable; in presence of downward-closed type,
there is no such general $\bot$ type (only one for
non-downward-closed types). We conjecture that the relaxed value restriction
is still sound in this case, because the covariance criterion is really used to
rule out mutable state rather than subtype from a $\bot$ type; but it
will be necessary to study the relaxation justification in more
details to formally establish this result.
\end{version}

\paragraph{Experiments with $v$-closure of type constructors as a new semantic property}

In a language with non-atomic subtyping such as OCaml, we need to
distinguish $v$-closed and non-$v$-closed type constructors. This is
a new semantic property that, in particular, must be reflected through
abstraction boundaries: we should be able to say about an abstract
type that it is $v$-closed, or not say anything.

How inconvenient in practice is the need to expose those properties to
have good variance for \GADTs? Will the users be able to determine
whether they want to enforce $v$-closure for a particular type they
are defining?

\begin{version}{\Not\Esop}
\paragraph{Experiments with subtyping constraints in \GADTs}

In \S\ref{sec/gadts-with-subtyping-constraints}, we have presented
a different way to define \GADTs with weaker constraints
(simple subtyping instead of equality) and stronger variance
properties. It is interesting to note that, for the few \GADTs
examples we have considered, using subtyping constraints rather than
equality constraints was sufficient for the desired applications of
the \GADT.

However, there are cases were the strong equality relying on
fine-grained closure properties is required. We need to consider more
examples of both cases to evaluate the expressiveness trade-off in,
for example, deciding to add only one of these solutions to an
existing type system.

On the implementation side, we suspect that adding
subtyping constraints to a type system that already supports \GADT and
private types should not require large engineering efforts
(in particular, it does not implies supporting the most general forms
of bounded polymorphism). Matching on a \GADT $\ty{\alpha}{t}$ already
introduces local type equalities of the form $\alpha = \app{T}{\bb}$ in
pattern matching clauses. Jacques Garrigue suggested that adding an
equality of the form $\alpha = \mc{private}~\app{T}{\bb}$ should
correspond to \GADT equations of the form $\alpha \leq \app{T}{\bb}$,
and lower bounds could be represented using the dual notion of
\code{invisible} types. Regardless of implementation difficulties, in
a system with only explicit subtyping coercion, such subtyping
constraints would still require more user annotations.
\end{version}

\begin{version}{\Not\Esop}
\paragraph{Mathematical structures for variance studies}

There has been work on more structured presentation of \GADTs as part
of a categorical framework (\cite{ghani-popl07} and
\cite{Hamana-Fiore-11}). This is orthogonal to the question of
variance and subtyping, but it may be interesting to re-frame the
current result in this framework.

Parametrized types with variance can also be seen as a sub-field of
order theory with very partial orders and functions with strong
monotonicity properties. Finally, we have been surprised to find that
geometric intuitions were often useful to direct our formal
developments. It is possible that existing work in these fields would
allow us to streamline the proofs, which currently are rather
low-level and tedious.
\end{version}

\paragraph{Completeness of variance annotations with domain information}
\label{future_work:domain_information}

The way we present \GADTs using equality constraints instead of the
codomain syntax is well-known to practictioners, under the form of
a ``factoring'' transformation where an arbitrary \GADT is expressed
as a simple \ADT, using the equality \GADT $\ty{(\alpha,\beta)}{eq}$
as part of the constructor arguments to reify equality information.

This transformation does not work anymore with our current notion of
\GADTs in presence of subtyping. Indeed, all we can soundly say about
the equality type $\ty{(\alpha,\beta)}{eq}$ is that it must be
invariant in both its parameters; using
$\ty{(\alpha,\app{T_i}{\bb})}{eq}$ as part of a constructor type would
force the paramter $\alpha$ to be invariant.

We think it would possible to re-enable factoring by \code{eq} by
considering \emph{domain information}, that is, information on
constraints that must hold for the type to be inhabited. If we
restricted the subtyping rule with conclusion $\ty{\bs}{t} \leq
\ty{\bs'}{t}$ to only cases where $\ty{\bs}{t}$ and $\ty{\bs'}{t}$ are
inhabited---with a separate rule to conclude subtyping in the
non-inhabited case---we could have a finer variance check, as we would
only need to show that the criterion \Rule{Seq} holds
between two instances of the inhabited domain, and not any
instance. If we stated that the domain of the type
$\ty{(\alpha,\beta)}{eq}$ is restricted by the constraint $\alpha =
\beta$, we could soundly declare the variance $\ty{(\virr\alpha,
  \virr\beta)}{eq}$ on this domain---which no longer prevents from
factoring out \GADTs by equality types.

\subsection*{Conclusion}

Checking the variance of \GADTs is surprisingly more difficult
(and interesting) than we initially thought.  We have studied a novel
criterion of upward and downward closure of type expressions and
proposed a corresponding syntactic judgment that is easily
implementable. We presented a core formal framework to prove both its
correctness and its completeness with respect to a natural semantic
criterion.

This closure criterion exposes important tensions in the design of
a subtyping relation, for which we previously knew of no convincing
example in the context of ML-derived programming languages. We have
suggested new language features to help alleviate these tensions,
whose convenience and practicality is yet to be assessed by real-world
usage.

Considering extensions of \GADTs in a rich type system is useful in
practice; it is also an interesting and demanding test of one's type
system design.


\def\urltilda{\kern -.15em\lower .7ex\hbox{\~{}}\kern .04em}
\def\urldot{\kern -.10em.\kern -.10em}
\def\urlhttp{http\kern -.10em\lower -.1ex\hbox{$\colon\!$}\kern -.12em\lower 0ex\hbox{/}\kern -.18em\lower 0ex\hbox{/}}

\bibliographystyle{alphaurl}
\bibliography{variance_gadts}

\begin{thebibliography}{EKRY06}

\bibitem[Abe06]{polarized-subtyping-for-sized-types}
Andreas Abel.
\newblock Polarized subtyping for sized types.
\newblock {\em Mathematical Structures in Computer Science}, 2006.
\newblock Special issue on subtyping, edited by Healfdene Goguen and Adriana
  Compagnoni.

\bibitem[EKRY06]{csharp-generalized-constraints}
Burak Emir, Andrew Kennedy, Claudio Russo, and Dachuan Yu.
\newblock Variance and generalized constraints for {C}\# generics.
\newblock In {\em Proceedings of the 20th European conference on
  Object-Oriented Programming}, ECOOP'06, 2006.

\bibitem[Gar04]{relaxing-the-value-restriction}
Jacques Garrigue.
\newblock Relaxing the value restriction.
\newblock In {\em In International Symposium on Functional and Logic
  Programming, Nara, LNCS 2998}, 2004.

\bibitem[KR05]{gadt-and-oop}
Andrew Kennedy and Claudio~V. Russo.
\newblock Generalized algebraic data types and object-oriented programming.
\newblock In {\em Proceedings of the 20th annual ACM SIGPLAN conference on
  Object-oriented programming, systems, languages, and applications}, 2005.
\newblock URL: \url{http://research.microsoft.com/pubs/64040/gadtoop.pdf}.

\bibitem[Pfe01]{pfenning:intextirr}
Frank Pfenning.
\newblock Intensionality, extensionality, and proof irrelevance in modal type
  theory.
\newblock In {\em 16th IEEE Symposium on Logic in Computer Science (LICS 2001),
  16-19 June 2001, Boston University, USA, Proceedings}, 2001.

\bibitem[SP07]{simonet-pottier-hmg-toplas}
Vincent Simonet and François Pottier.
\newblock A constraint-based approach to guarded algebraic data types.
\newblock {\em ACM Transactions on Programming Languages and Systems}, 29(1),
  January 2007.
\newblock URL: \url{http://doi.acm.org/10.1145/1180475.1180476}.

\bibitem[SR]{Scherer-Remy:gadts-meet-subtyping@long2012}
Gabriel Scherer and Didier R{\'e}my.
\newblock {GADTs} meet subtyping.
\newblock Long version, available electronically.
\newblock URL: \url{http://gallium.inria.fr/~remy/gadts/}.

\end{thebibliography}

\begin{version}{\Not\Esop}

\pagebreak
\appendix
\tableofcontents
\end{version}

\end{document}